\documentclass{article}

\usepackage{amsfonts}
\usepackage{amssymb}
\usepackage{mathrsfs}
\usepackage{amstext}
\usepackage{graphicx}
\usepackage{color}
\usepackage{url}

\usepackage{bbm}

\usepackage{bm}
\usepackage{amsmath}

\usepackage{theorem}
\usepackage{amssymb}
\usepackage{lastpage}
\usepackage{fancyhdr}
\usepackage{fancybox}
\font\msbm=msbm10

\theoremstyle{plain}
\newtheorem{theorem}{theorem}
\newtheorem{lemma}[theorem]{Lemma}

\newtheorem{proposition}[theorem]{Proposition}

\newtheorem{remark}[theorem]{Remark}
\def\mathbb#1{\hbox{\msbm{#1}}}

\newcommand{\Tr}{ {\rm Tr}}

\newcommand{\N}{{\mathbb{N}}}
\newcommand{\R}{{\mathbb{R}}}
\newcommand{\Z}{{\mathbb{Z}}}
\newcommand{\C}{{\mathbb{C}}}

\renewcommand{\P}{{\mathbb{P}}}
\newcommand{\E}{{\mathbb{E}}}

\newcommand{\sgn}{\operatorname{sgn}}
\newcommand{\sparsity}{s}

\newcommand{\n}{n}				
\newcommand{\m}{m}				

\newcommand{\beq}{\begin{eqnarray}}
\newcommand{\eeq}{\end{eqnarray}}
\newcommand{\beqn}{\begin{eqnarray*}}
\newcommand{\eeqn}{\end{eqnarray*}}

\newcommand{\supp}{\operatorname{supp}}

\newcommand{\nns}[1]{{|\!|\!|}{#1}{|\!|\!|}_s}

\renewcommand{\Re}{\operatorname{Re}}
\renewcommand{\Im}{\operatorname{Im}}

\newcommand{\transk}{k}

\newcommand{\vct}[1]{\bm{#1}}
\newcommand{\mtx}[1]{\bm{#1}}

\newcommand{\qed}{\rule{2.5mm}{2.5mm}}
\newenvironment{proof}{\noindent
{\bf\underline{Proof:} }}
{\hspace*{\fill}\qed\vskip1em}

 \newcommand{\mymatrixNewNewSPECIALB}[1]{
            \left(\footnotesize
                \setlength{\arraycolsep}{4pt}
                \renewcommand{\arraystretch}{.7}
                    \begin{array}{cccc|ccc|cc} #1
                    \end{array}
            \renewcommand{\arraystretch}{1}
    \right)    }

    \newcommand{\mymatrixNewNewSPECIAL}[1]{
            \left(\footnotesize
                \setlength{\arraycolsep}{4pt}
                \renewcommand{\arraystretch}{.7}
                    \begin{array}{ccccc|ccccc|cc} #1
                    \end{array}
            \renewcommand{\arraystretch}{1}
    \right)    }

\begin{document}

\title{The restricted isometry property for time-frequency structured random matrices}
\author{G\"otz E. Pfander\footnotemark[1], Holger Rauhut\footnotemark[2],  Joel A. Tropp\footnotemark[3]}


\maketitle

\begin{center}
Dedicated to Hans Georg Feichtinger on the occasion of his 60th birthday.
\end{center}


\begin{abstract} We establish the restricted isometry property for finite dimensional Gabor systems,
that is, for families of time--frequency shifts of a randomly chosen window function. We show that the $\sparsity$-th order
restricted isometry constant of the associated $n \times n^2$ Gabor synthesis matrix is small provided $s \leq c \, n^{2/3} / \log^2 n$.
This improves on previous estimates that exhibit quadratic scaling of $n$ in $\sparsity$. Our proof develops bounds for
a corresponding chaos process.
\end{abstract}
\vspace{0.5cm}
\begin{tabbing}
time-frequency analysis, random matrix, chaos process.
\end{tabbing}

\renewcommand{\thefootnote}{\fnsymbol{footnote}}

\footnotetext[1]{GEP is with School of Engineering and Science, Jacobs
University Bremen, 28759 Bremen, Germany, (e-mail: \url{g.pfander@jacobs-university.de}).}

\footnotetext[2]{HR is with Hausdorff Center for Mathematics and Insitute for Numerical Simulation, University of Bonn,
Endenicher Allee 60, 53115 Bonn, Germany (e-mail: \url{rauhut@hcm.uni-bonn.de}). 
}

\footnotetext[3]{JAT is with California Institute of Technology, Pasadena, CA 91125
USA (e-mail: \url{jtropp@cms.caltech.edu}). 
}

\section{Introduction and statements of results}

Sparsity has become a key concept in applied mathematics and engineering.  
This is largely due to the empirical observation
that a large number of real-world signals can be represented 
well by a sparse expansion in an appropriately chosen system of basic signals.
Compressive sensing~\cite{ca06-1,carota06,cata06,do06-2,fora11,ra10}
predicts that a small number of linear samples suffices to capture all the information in a sparse vector and that, furthermore, we can recover the sparse vector from these samples using efficient algorithms.  This discovery has a number of potential applications in signal processing, as well as other areas of science and technology.

Linear data acquisition is described by a measurement matrix.  The \emph{restricted isometry property} (RIP)~\cite{carota06-1,cata06,fora11,ra10} is by-now a standard tool for studying how efficiently the measurement matrix captures information about sparse signals.  The RIP also streamlines the analysis of signal reconstruction algorithms, including $\ell_1$-minization, greedy and iterative algorithms.
Up to date there are no deterministic constructions of measurement matrices available that satisfy the RIP with the optimal scaling behavior; see, for example, the discussions in \cite[Sec.~2.5]{ra10} and \cite[Sec.~5.1]{fora11}.  In contrast, a variety of random measurement matrices exhibit the RIP with optimal scaling, including Gaussian matrices and Rademacher matrices~\cite{badadewa08,dota09,rascva08,cata06}.

Although Gaussian random matrices are optimal for sparse recovery \cite{do06-2,foparaul10}, they have limited use in practice because many applications impose structure on the matrix.  Furthermore, recovery algorithms are significantly more efficient when the matrix admits a fast matrix--vector multiplication.  For example, random sets of rows from a discrete Fourier transform matrix model the measurement process in MRI imaging and other applications. These random partial Fourier matrices lead to fast recovery algorithms because they can utilize the FFT.  It is known that a random partial Fourier matrix satisfies 
a near-optimal RIP~\cite{cata06,ruve08,ra08,ra10} with high probability;
see also \cite{ra10,rawa10} for some generalizations. 

This paper studies another type of structured random matrix that arises from time-frequency analysis, and has potential
applications for the channel identification problem \cite{pfrata08} in wireless communications and sonar \cite{mi87,st99},
as well as in radar \cite{hest09}.
The columns of the considered $n \times n^2$ matrix consist of 
all discrete time-frequency shifts of a random vector. Previous analysis of this matrix
has provided bounds for the coherence \cite{pfrata08}, as well as nonuniform sparse recovery guarantees using $\ell_1$-minimization
\cite{pfra10}. However, the so far best available bounds on the restricted isometry constants were derived from coherence bounds \cite{pfrata08} and, therefore,  exhibit
highly non-optimal quadratic scaling of $n$ in the sparsity $\sparsity$. This paper dramatically improves on these bounds.
Such an improvement is important because the nonuniform recovery guarantees in \cite{pfra10} 
apply only for $\ell_1$-minimization, they do not provide
stability of reconstruction, and they do not show the existence of a single time-frequency structured measurement matrix that is able
to recover all sufficiently sparse vectors. Also it is of theoretical interest whether Gabor systems, that is, the columns of
our measurement matrix, can possess the restricted isometry property.
Nevertheless, our results still fall short of the optimal scaling that one might hope for.

Our approach is similar to the recent restricted isometry analysis for partial random circulant matrices in \cite{rarotr10}. Indeed,
also here we bound a chaos process of order $2$, by means of a Dudley type inequality for such processes due to Talagrand
\cite{ta05-2}. This requires to estimate covering numbers of the set of unit norm $\sparsity$-sparse
vectors with respect to two different metrics induced by the process. In contrast to \cite{rarotr10}, the specific structure 
of our problem does not allow us to reduce to the Fourier case, and to apply covering number estimates shown in \cite{ruve08}.

This paper is organized as follows. In Section~\ref{subsection:compressedsensing} we recall central concepts in compressive sensing. Section~\ref{section:time-frequencystructured}  introduces the time-frequency structured measurement matrices that are considered in this paper, and we state our main result, Theorem~\ref{theorem:main}. Remarks on applications in wireless communications and radar, as well as the relation of this paper to previous work are given in Sections~\ref{section:relationpreviouswork} and \ref{section:applicationwireless}, respectively.
Sections~\ref{section:expectation}, \ref{sec:cover:proofs} and \ref{section:probability}
provide  the proof of Theorem~\ref{theorem:main}.

\subsection{Compressive Sensing}\label{subsection:compressedsensing}

In general, reconstructing $\vct{x} = (x_1, \dots, x_N)^T \in \C^N$ from
\begin{equation}\label{eqn:linearsystem}
  \vct{y} ~=~ \mtx{A} \vct{x}\in \C^n,
\end{equation}
where  $\mtx{A} \in \C^{n \times N}$ and $n \ll N$ (in this paper, we  have $N=n^2$) is impossible without substantial {\it a-priori} information on $\vct{x}$.  In compressive sensing the assumption that $\vct{x}$ is $\sparsity$-sparse,  that is, $\|\vct{x}\|_0 := \#\{\ell : x_\ell \neq 0\} \leq \sparsity$ for some $\sparsity \ll N$ is introduced to ensure uniqueness and efficient recoverability of $\vct{x}$.
More generally, under the assumption that $\vct{x}$ is well-approximated by a sparse vector, the question is posed whether an optimally sparse approximation to $\vct{x}$ can be found efficiently.

Reconstruction of a sparse vector $\vct{x}$ by means of the $\ell_0$-minimization problem,
\[
\min_{\vct{z}} \|\vct{z}\|_0 \quad \mbox{ subject to }\quad \vct{y} ~=~ \mtx{A} \vct{z},
\]
is {\sf NP}-hard \cite{na95} and therefore not tractable. Consequently, a number of alternatives to $\ell_0$-minimization, for example, greedy algorithms \cite{blda09,fo10-2,neve09-1,netr08,tr04},
have been proposed in the literature. The most popular approach utilizes $\ell_1$-minimization \cite{carota06,chdosa99,do06-2}, that is, the convex program
\begin{equation}\label{l1:min}
\min_{\vct{z}} \|\vct{z}\|_1 \quad \mbox{ subject to } \vct{y} ~=~ \mtx{A} \vct{z},
\end{equation}
is solved, where $\|\vct{z} \|_1=|z_1|+|z_2|+\ldots+|z_N|$ denotes the usual $\ell_1$ vector norm.

To guarantee recoverability of the sparse vector $\vct{x}$ in \eqref{eqn:linearsystem} by means of $\ell_1$-minimization and greedy algorithms, it suffices to establish the restricted isometry property (RIP) of the so-called measurement matrix $\mtx{A}$:
define the restricted isometry constant $\delta_{\sparsity}$ of an $n \times N$ matrix $\mtx{A}$ to be the smallest positive number that satisfies
\begin{equation}
	\label{eq:RIPs}
	(1-\delta_\sparsity)\|\vct{x}\|^2_2 ~\leq~ \|\mtx{A} \vct{x}\|^2_2 ~\leq~ (1+\delta_\sparsity)\|\vct{x} \|^2_2
	\quad
	\text{ for all } \vct{x} \mbox{ with } \|\vct{x}\|_0\leq\sparsity.
\end{equation}
In words, the statement~\eqref{eq:RIPs} requires that all column submatrices of $\mtx{A}$ with at most $\sparsity$ columns are well-conditioned.
Informally, $\mtx{A}$ is said to satisfy the RIP with order $\sparsity$ when $\delta_{\sparsity}$ is ``small''. 

Now, if the matrix $\mtx{A}$ obeys \eqref{eq:RIPs} with
\begin{equation}\label{delta:kappa}
\delta_{\kappa \sparsity} ~<~ \delta^*
\end{equation}
for suitable  constants $\kappa \geq 1$ and $\delta^* < 1$, then many algorithms precisely recover any $\sparsity$-sparse vectors $\vct{x}$ from the measurements $\vct{y} = \mtx{A} \vct{x}$.  Moreover, if  $\vct{x}$ can be well approximated by  an $\sparsity$ sparse vector, then for noisy observations
\[
\vct{y} ~=~ \mtx{A} \vct{x} + \vct{e}
\quad\text{where}\quad
\|\vct{e}\|_2 \leq \tau,
\]
these algorithms return a reconstruction $\widetilde{\vct{x}}$ that satisfies an error bound of the form
\begin{equation}\label{rec:stability}
\| \vct{x} - \widetilde{\vct{x}} \|_2~\leq~ C_1 \frac{\sigma_\sparsity(\vct{x})_1}{\sqrt{\sparsity}} + C_2 \tau,
\end{equation}
where $\sigma_\sparsity(\vct{x})_1 = \inf_{\|\vct{z}\|_0 \leq \sparsity} \|\vct{x} - \vct{z}\|_1$ denotes the error of best $\sparsity$-term approximation in $\ell_1$ and $C_1, C_2$ are positive constants. For illustration, we include Table \ref{tableRIP} which lists 
available values  
for the constants $\kappa$ and $\delta^*$ in \eqref{delta:kappa} that guarantee \eqref{rec:stability} for several algorithms along with respective references.

\begin{table}[htdp]
\begin{center}
\begin{tabular}{|l|l|c|l|}
\hline
Algorithm & $\kappa$ & $\delta^*$ & References\\
\hline
$\ell_1$-minimization \eqref{l1:min} & $2$ & $\frac{3}{4+\sqrt{6}} \approx 0.4652$ & \cite{cawaxu10,ca08,carota06-1,fo10-3}\\
CoSaMP & $4$ & $\sqrt{\frac{2}{5 + \sqrt{73}}} \approx 0.3843$& \cite{fo10,netr08}\\
Iterative Hard Thresholding & $3$ & $1/2$& \cite{blda09,fo10-3}\\
Hard Thresholding Pursuit & $3$ & $1/\sqrt{3} \approx 0.5774$& \cite{fo10-2}\\
\hline
\end{tabular}
\caption{Values of the constants $\kappa$ and $\delta^*$ in \eqref{delta:kappa} that guarantee success for  various recovery algorithms.}
\label{tableRIP}
\end{center}
\end{table}%

For example, Gaussian random matrices, that is, matrices that have independent, normally distributed entries with mean zero and variance one,  have been shown \cite{badadewa08,cata06,mepato09} to have restricted isometry
constants of $\frac{1}{\sqrt{n}} \mtx{A}$ satisfy $\delta_\sparsity \leq \delta$ with high probability provided that
\[
n ~\geq~ C \delta^{-2} \sparsity \log(N/\sparsity).
\]
That is, the number $n$ of Gaussian measurements required to reconstruct an $\sparsity$-sparse signal of length $N$ is \emph{linear} in the sparsity and \emph{logarithmic} in the ambient dimension.
See \cite{badadewa08,cata06,mepato09,fora11,ra10} for precise statements and extensions to Bernoulli and subgaussian matrices.
It follows from lower estimates of Gelfand widths that this bound on the required
samples is optimal \cite{codade09,foparaul10,gagl84}, that is, the $\log$-factor must be present.

As discussed above, no deterministic construction of a measurement matrix is known which provides RIP with optimal scaling of the recoverable sparsity $\sparsity$ in the number of measurements $n$. 
In fact, all available proofs of the RIP with close to optimal scaling require  the measurement matrix to contain some randomness.
In Table~\ref{table:entropy} we list the Shannon entropy (in bits) of various random matrices along with the available RIP estimates. 
Compared to Gaussian random matrices, the Gabor synthesis measurement matrices constructed in this paper introduces only a small amount of randomness, that is, the presented measurement matrix depends only on the so-called Gabor window, a random vector of length $n$, which can be chosen to be a normalized copy of a Rademacher vector. 
Moreover, the random Gabor matrix provably provides scaling of $\sparsity$ roughly in $n^{2/3}$, which significantly improves on known deterministic constructions. Clearly, such scaling falls short of the optimal one, but we expect that
it is possible to establish linear scaling of $\sparsity$ in $n$ up to $\log$-factors, similar to Gaussian matrices or partial random Fourier matrices. However,
such improvement seems to require more powerful methods to estimate chaos processes than presently available.

\begin{table}[htdp]
\begin{center}
\begin{tabular}{|l|l|l|l|}
\hline
$n\times N$ Measurement matrix & Shannon entropy & RIP regime & References\\
\hline
Gaussian 
& $nN\,\frac 1 2 \log (2\pi e)$ & $s \leq C n / \log N$ & \cite{badadewa08,dota09,ruve08}\\
Rademacher entries & $nN $ & $s \leq C n / \log N$ & \cite{badadewa08}\\
Partial Fourier matrix & $N \log_2 N {-} n \log_2 n $ & $s \leq C n / \log^4 N$ & \cite{rarotr10,ruve08}\\
  & ${-} (N{-}n) \log_2 (N{-}n)$ &  & \\
Partial circulant Rademacher  & $N$ & $s \leq C n^{2/3} / \log^{2/3} N$ & \cite{rarotr10}\\
Gabor, Rademacher window & $n$ & $s \leq C n^{2/3} / \log^2 n$ & this paper \\
Gabor, Alltop window & 0 & $s \leq C \sqrt{n}$ & \cite{pfrata08}\\
\hline
\end{tabular}
\caption{List of measurement matrices that have been proven to be RIP, scaling of sparsity $s$ in the number of measurements $n$, and the respective Shannon entropy of the (random) matrix. }
\label{table:entropy}
\end{center}
\end{table}%

\subsection{Time-frequency structured measurement matrices}\label{section:time-frequencystructured}

In this paper, we provide probabilistic estimates of the restricted isometry constants for matrices whose columns
are time--frequency shifts of a randomly chosen vector.  To define these matrices, we let  $\mtx{T}$ denote the cyclic shift,
also called translation operator, and $\mtx{M}$ the modulation operator, or frequency shift operator, on $\C^n$. They are
defined by
 \begin{equation}\label{eq:trans_mod}
 (\mtx{T}
\vct{h})_q=h_{q\ominus 1}\quad\mbox{and}\quad (\mtx{M}\vct{h})_q = e^{2\pi i
q/n} h_q = \omega^q h_q,
\end{equation}
where $\ominus$ is subtraction modulo $\n$ and $\omega=e^{2\pi i/n}$.
Note that
\begin{equation}\label{eq:trans_mod2}
(\mtx{T}^\transk
\vct{h})_q=h_{q\ominus \transk}\quad\mbox{and}\quad (\mtx{M}^{\ell}\vct{h})_q = e^{2\pi i\ell
q/n} h_q = \omega^{\ell
q} h_q.
\end{equation}
The operators $\mtx{\pi}(\lambda) = \mtx{M}^\ell \mtx{T}^\transk$, $\lambda = (\transk,\ell)$,
are called time-frequency shifts and the system $\{\mtx{\pi}(\lambda): \lambda \in
\Z_n{\times}\Z_n\}$, $\Z_n = \{0,1,\ldots,n-1\}$, of all time-frequency shifts
forms a basis of the matrix space $\C^{n{\times} n}$ \cite{LPW05,krpfra07}.

We choose $\vct{\epsilon} \in \C^{n}$ to be a Rademacher or Steinhaus sequence, that is,
a vector of independent random variables
taking the values $+1$ and $-1$ with equal probability, respectively taking values uniformly distributed on the complex torus
$S^1 = \{z \in \C, |z| = 1\}$.
The normalized window is
\[
\vct{g}=n^{-1/2}\vct{\epsilon},
\]
and the set
\begin{equation}\label{GaborSystem}
\{\mtx{\pi}(\lambda) \vct{g}: \lambda \in \Z_n{\times} \Z_n\}
\end{equation}
is called a full Gabor system with window $\vct{g}$ \cite{Gro01}. The matrix $\mtx{\Psi}_{\vct{g}} \in \C^{n
{\times} n^2}$ whose columns list the members $\mtx{\pi}(\lambda) \vct{g}$, $\lambda
\in \Z_n{\times} \Z_n$, of the Gabor system is referred to as Gabor synthesis matrix \cite{Chr03,LPW05,PR09}. Note that
$\mtx{\Psi}_{\vct{g}}$ allows for fast matrix vector multiplication algorithms based on the FFT.
The main result of this paper addresses the restricted isometry constants of $\mtx{\Psi}_{\vct{g}}$.
Below $\E$ denotes expectation and $\P$ the probability of an event.

\begin{theorem}\label{theorem:main} Let $\mtx{\Psi}_{\vct{g}} \in \C^{n \times n^2}$ be a draw of the random Gabor synthesis matrix
with normalized Steinhaus or Rademacher generating vector.
\begin{enumerate}
\item[(a)]  The expectation of the restricted isometry constant $\delta_s$ of $\mtx{\Psi}_{\vct{g}}$, $s\leq n$, satisfies
  \begin{equation}
     \E\, \delta_\sparsity \leq \max\Big\{C_1 \sqrt{\frac{\sparsity^{3/2}}{n}} \log\sparsity\sqrt{ \log n}, \
C_2 \frac{\sparsity^{3/2}\log^{3/2} n}{n} \Big\} ,
\end{equation}
where $C_1,C_2> 0$ are universal constants.
\item[(b)]
For $0\leq \lambda\leq 1$, we have
\begin{equation}
    \P(\delta_s \geq \E[\delta_s] +\lambda) \leq e^{-\lambda^2/\sigma^2},\quad \text{ where } \sigma^2= \frac{C_3 s^{\frac 3 2} \log n \, \log^2s}{n}
\end{equation}
with $C_3>0$ being a universal constant.
\end{enumerate}
\end{theorem}

With slight variations of the proof one can show similar statements for normalized Gaussian or subgaussian
random windows $\vct{g}$.

Roughly speaking $\mtx{\Psi}_{\vct{g}}$ satisfies the RIP of order $\sparsity$ with high
probability if $n \geq C s^{3/2} \log^{3}(n)$, or equivalently if,
\[
\sparsity \leq c n^{2/3} / \log^{2} n.
\]
We expect that this is not the optimal estimate, but improving on this seems to require
more sophisticated techniques than pursued in this paper. There are known examples \cite{leta91,ta05-2} for
which the central tool in this paper, the Dudley type inequality for chaos processes stated in Theorem~\ref{Dudley:chaos},
is not sharp.  We may well be  facing one of these cases here.

Numerical tests illustrating the use of $\mtx{\Psi}_{\vct{g}}$ for compressive sensing are presented in \cite{pfrata08}.
They illustrate that empirically $\mtx{\Psi}_{\vct{g}}$ performs very similarly to a Gaussian matrix.

\subsection{Application in wireless communications and radar}\label{section:applicationwireless}

An important task in wireless communications is to identify the communication channel at hand,
 that is, the channel opperator, by probing it
with a small number of known transmit signals; ideally a single probing signal. A common finite-dimensional model for the channel operator, that combines digital (discrete) to analog conversion, the analog channel, and analog to digital conversion. It is given by
\cite{be63,Cor01,grpf08,Pae01}
\[
 \mtx{\Gamma} = \sum_{\lambda \in \Z_n{\times} \Z_n} x_\lambda \mtx{\pi}(\lambda).
\]
Time-shifts model delay due to multipath-propagation, while frequency-shifts model the Doppler effect due to moving transmitter, receiver, and/or scatterers.
Physical considerations often suggest that $\vct{x}$ is rather sparse as, indeed, the number of present scatterers can be assumed to be small in most cases.
The same model is used as well in sonar \cite{mi87,st99}  and radar \cite{hest09}.

Our task is to identify from a single input output pair $(\vct{g},\mtx{\Gamma} \vct{g})$ the coefficient vector $x$.
In other words, we need to reconstruct
$\mtx{\Gamma} \in \C^{n{\times} n}$, or equivalently $\vct{x}$, from its action $\vct{y}= \mtx{\Gamma} \vct{g}$ on a single vector
$\vct{g}$. Writing
\begin{equation}\label{equation:matrixIdent}
  \vct{y} = \mtx{\Gamma} \vct{g} \,=\, \sum_{\lambda \in \Z_n {\times} \Z_n} x_\lambda \mtx{\pi}(\lambda) \vct{g}
  = \mtx{\Psi}_{\mtx{g}} \vct{x}
\end{equation}
with unknown but sparse $\vct{x}$, we arrive at a compressive sensing problem. In this setup, we clearly have the freedom to choose $\vct{g}$,
and we may choose it as a random Rademacher or Steinhaus sequence.
Then the restricted isometry property of
$\mtx{\Psi}_{\vct{g}}$, as shown in Theorem \ref{theorem:main}, ensures recovery of sufficiently sparse $\vct{x}$, and hence, of the
associated operator $\mtx{\Gamma}$.

Recovery of the sparse $\vct{x}$ in \eqref{equation:matrixIdent} can also be interpreted as finding a sparse time-frequency representation of a given $\vct{y}$ with respect to the window $\vct{g}$. From an application point of view though, the vectors considered here are not well suited to describe meaningful sparse time-frequency representations of $\vct{x}$ as all $\vct{g}$ that are known to guarantee RIP of $\mtx{\Psi}_{\vct{g}}$ are very poorly localized both in time and in frequency.

\subsection{Relation with previous work}\label{section:relationpreviouswork}

Time-frequency structured matrices $\mtx{\Psi}_{\vct{g}}$ appeared in the study of frames with (near-)optimal coherence. Recall that
the coherence of a matrix $\mtx{A} = (\vct{a}_1|\ldots|\vct{a}_N)$ with normalized columns $\|\vct{a}_\ell\|_2=1$ is defined as
\[
\mu := \max_{\ell \neq k} |\langle \vct{a}_\ell,\vct{a}_k\rangle|.
\]
Choosing the Alltop window \cite{al80,hest03} $\vct{g} \in \C^n$ with entries $g_\ell = n^{-1/2} e^{2 \pi i \ell^3/n}$ for $n\geq 5$ prime yields $\mtx{\Psi}_{\vct{g}}$  with coherence
\[
\mu = \frac{1}{\sqrt{n}}.
\]
Due to the general lower bound $\mu \geq \sqrt{\frac{N-n}{n(N-1)}}$ for an $n \times N$ matrix \cite{hest03},
this coherence is almost optimal.
Together with the bound $\delta_{\sparsity} \leq (s-1)\mu$ we obtain
\[
\delta_{\sparsity} \leq \frac{\sparsity-1}{\sqrt{n}}.
\]
This requires a scaling $\sparsity \leq c \sqrt{n}$ to achieve sufficiently small RIP and sparse recovery, which clearly is worse than the main result of this paper.

The coherence of $\mtx{\Psi}_{\vct{g}}$ with Steinhaus sequence $\vct{g}$ is estimated in \cite{pfrata08} by
\[
\mu \leq c \sqrt{\frac{\log(n/\varepsilon)}{n}},
\]
holding with probability at least $1-\varepsilon$. As before, this does not give better than quadratic scaling of $n$ in $\sparsity$
in order to have small RIP constants $\delta_s$.

The following nonuniform recovery results for $\ell_1$-minimization with $\mtx{\Psi}_{\vct{g}}$ and Steinhaus sequence $\vct{g}$
was derived in \cite{pfra10}.
\begin{theorem}\label{thm:nonuniform} Let $\vct{x} \in \C^n$ be $\sparsity$-sparse. Choose a Steinhaus sequence $\vct{g}$ at random. Then with probability at least $1-\varepsilon$, the vector $\vct{x}$ can be recovered from $\vct{y} = \mtx{\Psi}_{\vct{g}}\vct{x}$ via $\ell_1$-minimization provided
\[
\sparsity \leq c \frac{n}{\log(n/\varepsilon)}.
\]
\end{theorem}
Clearly, the (optimal) almost linear scaling of $n$ in $\sparsity$ of this estimate is better than the RIP estimate of
the main Theorem \ref{theorem:main}. However, the conclusion is weaker than what can be derived
using the restricted isometry property: recovery in Theorem~\ref{thm:nonuniform} is nonuniform in the sense that a given $\sparsity$-sparse vector
can be recovered with high probability from a random draw of the matrix $\mtx{\Psi}_{\vct{g}}$. It is not stated that a single matrix
$\mtx{\Psi}_{\vct{g}}$ can recover all $\sparsity$-sparse vectors simultaneously. Moreover, nothing is said about the stability
of recovery, while in contrast,  small RIP constants imply \eqref{rec:stability}.
Therefore, our main Theorem \ref{theorem:main} is of high interest and importance, despite the better scaling
in Theorem \ref{thm:nonuniform}. Moreover, we expect that an improvement of the RIP estimate is possible, although
it is presently not clear how this can be achieved.

Partial random circulant matrices are a different, but closely related measurement matrix, studied in \cite{bahanora10,ra09,ra10,rarotr10}. They model convolution
with a random vector followed by subsampling on an arbitrary (deterministic) set. The so far best estimate of the
restricted isometry constants $\delta_\sparsity$ of such an $n \times N$ matrix in \cite{rarotr10}
requires $n \geq c (s \log N)^{3/2}$, similarly to the main result of this paper. The corresponding analysis requires
to bound as well a chaos process, which is also achieved by the Dudley type bound of Theorem \ref{Dudley:chaos} below.
Nonuniform recovery guarantees for partial random circulant matrices similarly to Theorem \ref{thm:nonuniform}
are contained in \cite{ra09,ra10}. The analysis of circulant matrices benefits from a simplified arithmetic in the Fourier domain, a tool not available to us in the case of Gabor synthesis matrices.  Hence, the analysis presented here is more involved.

\section{Expectation of the restricted isometry constants}\label{section:expectation}

We first estimate the expectation of the restricted isometry constants of the
random Gabor synthesis matrix, that is, we shall prove Theorem~\ref{theorem:main}(a).
To this end, we first rewrite the restricted isometry
constants $\delta_\sparsity$. Let $T = T_\sparsity = \{\vct{x} \in \C^{n^2},
\|\vct{x}\|_2 = 1, \|\vct{x}\|_0 \leq \sparsity\}$. Introduce the following semi-norm on
Hermitian matrices $A$,
$$
\nns{\mtx{A}}  = \sup_{\vct{x} \in T_\sparsity} |\vct{x}^* \mtx{A} \vct{x}|.
$$
Then the restricted isometry constants of $\mtx{\Psi}=\mtx{\Psi}_{\vct{g}}$ can be written as
$$
\delta_{\sparsity} = \nns{\mtx{\Psi}^* \mtx{\Psi} - \mtx{I}},
$$
where $\mtx{I}$ denotes the identity matrix.
Observe that the Gabor synthesis matrix $\mtx{\Psi}_{\vct{g}}$ takes the form
            $$
\mtx{\Psi}_{\vct{g}}=
\mymatrixNewNewSPECIALB{
            g_0& g_{n-1}   &\cdots & g_{1} & g_0  &\cdots & g_{1}& \quad \cdots   &  g_{1}  \\
             g_1& g_0   & \cdots  &     g_2 & \omega g_1 & \cdots  & \omega g_2 & \quad  \cdots  &  \omega^{n-1}g_{2}   \\
            g_2& g_1 &\cdots & g_3 &   \omega^2g_2 &\cdots & \omega^2 g_3  &  \quad \cdots  &  \omega^{2(n-1)}g_{3}   \\
            g_3& g_2 &\cdots & g_4 &   \omega^3g_3 & \cdots & \omega^3 g_4  &  \quad \cdots  &  \omega^{3(n-1)}g_{4}   \\
            \vdots &  \vdots & \ddots & \vdots &\vdots & \ddots & \vdots  &  &     \vdots \\
            g_{n-1}&g_{n-2} &\cdots  & g_0 &\omega^{n-1}  g_{n-1} &\cdots  & \omega^{n-1}g_0 &  \quad \cdots  &  \omega^{(n-1)^2}g_{0}        }\,
            .$$
Our analysis in this section employs the representation
$$\mtx{\Psi}_{\vct{g}}=\sum_{q=0}^{n-1} g_q \, \mtx{A}_q$$ with
\begin{eqnarray}
\mtx{A}_0&=& \mymatrixNewNewSPECIAL{
            1& 0 &0  &\cdots & 0 & 1& 0 &0  &\cdots & 0& \quad \cdots   &  0  \\
             0& 1 &0  & \cdots  &     0 & 0&\omega   &0  &  \cdots  &0& \quad \cdots  & 0   \\
            0& 0&1 &\cdots & 0 &  0&0&\omega^2  &\cdots & 0 &  \quad \cdots  &0   \\
            \vdots & \vdots & \vdots & \ddots & \vdots &\vdots & \vdots & \vdots & \ddots & \vdots  &  &     \vdots \\
            0&0&0 &\cdots  & 1 &0&0&0&\cdots  & \omega^{n-1} &  \quad \cdots  &  \omega^{(n-1)^2}        }\,\notag\\
            &=& \big(\mtx{I} \big|\mtx{M}\big|\mtx{M}^2\big| \cdots \big|\mtx{M}^{n-1}\big),\notag  \\
\mtx{A}_1&=&
\mymatrixNewNewSPECIAL{
            0& 0 &0  &\cdots & 1 &0& 0 &0  &\cdots & 1& \quad \cdots   & 1 \\
             1& 0 &0  & \cdots  &     0 & \omega &0 &0  &\cdots  &0 &  \quad  \cdots  &  0  \\
           0& 1&0 &\cdots & 0& 0&\omega^2 & 0 &\cdots &0 &  \quad \cdots  &0   \\
            \vdots & \vdots & \vdots & \ddots & \vdots &\vdots & \vdots & \vdots & \ddots & \vdots  &  &     \vdots \\
            0&0&0&\cdots  & 0 &0&0&0&\cdots  & 0 &  \quad \cdots  & 0        }\,\notag\\
           & = & \big(\mtx{T}\big|\mtx{MT}\big|\mtx{M}^2\mtx{T}\big| \cdots \big|\mtx{M}^{n-1}\mtx{T}\big),\notag
\end{eqnarray}
and so on. In short, for $q\in\Z_n$,
\begin{equation}
  \mtx{A}_q          = \big(\mtx{T}^q\big|\mtx{M}\mtx{T}^q\big|\mtx{M}^2\mtx{T}^q\big| \cdots \big|\mtx{M}^{n-1}\mtx{T}^q\big). \label{eq:def-Aq}
\end{equation}
Observe that
$$
\mtx{H} := \mtx{\Psi}^*\mtx{\Psi} - \mtx{I} = - \mtx{I} +  \frac{1}{n} \sum_{q,q' = 0}^{n-1}  \overline{\epsilon_{q'}} \epsilon_q  \, \mtx{A}^\ast_{q'} \mtx{A}_q\,.
$$
Using \eqref{equation:sumWqqIsId} below, it follows that
\begin{equation}\label{def:H}
\mtx{H} = \frac{1}{n} \sum_{q'\neq q}\overline{\epsilon_{q'}}\, \epsilon_q\, \mtx{A}^\ast_{q'} \mtx{A}_q
= \frac 1 n  \sum_{q', q}\overline{\epsilon_{q'}}\, \epsilon_q\, \mtx{W}_{q',q},
\end{equation}
where, for notational simplicity, we use here and in the following  $ \mtx{W}_{q',q}=
\mtx{A}^\ast_{q'} \mtx{A}_q$ for $q\neq q'$ and $ \mtx{W}_{q',q}=0 $ for $q= q'$.
We
employ the matrix $\mtx{B}(\vct{x})\in \C^{n\times n}$, $\vct{x}\in T_\sparsity$, given by matrix entries
\begin{equation}
  B(\vct{x})_{q',q}=\vct{x}^\ast \mtx{W}_{q',q} \vct{x}.\label{eq:def-B(x)}
\end{equation}
Then we have
\begin{equation}\label{eqn:expect:delta}
    n\,\E\delta_\sparsity =\E \sup_{\vct{x}\in T_\sparsity} |Y_{\vct{x}} |=\E \sup_{\vct{x}\in T_\sparsity} |Y_{\vct{x}} - Y_{\vct{0}} |\,,
\end{equation}
where
\begin{equation}\label{chaos:process}
Y_{\vct{x}} \,=\, \vct{\epsilon}^\ast \mtx{B}(\vct{x})\vct{\epsilon}= \sum_{q'\neq q}\overline{\epsilon_{q'}}\, \epsilon_q\,  \vct{x}^*\mtx{A}^\ast_{q'}
\mtx{A}_q \vct{x}
\end{equation}
and $\vct{x}\in T_\sparsity= \{\vct{x} \in \C^{n\times n}, \|\vct{x}\|_2 \leq 1, \|\vct{x}\|_0 \leq
\sparsity\}$.
A process of the type \eqref{chaos:process} is called Rademacher or Steinhaus chaos process of order $2$.
In order to bound such a process, we use the following 
Theorem, see for example, \cite[Theorem 11.22]{leta91} or \cite[Theorem 2.5.2]{ta05-2},
where it is stated
for Gaussian processes and in terms of majorizing measure (generic chaining) conditions. The formulation below requires
the operator norm $\|\mtx{A}\|_{2 \to 2} = \max_{\|\vct{x}\|_2=1} \|\mtx{A}\vct{x}\|_2$ and the Frobenius norm
$\|\mtx{A}\|_F = \Tr (\mtx{A}^* \mtx{A})^{1/2} = (\sum_{j,k} |A_{j,k}|^2)^{1/2}$, where $\Tr(\mtx{A})$
denotes the trace of a matrix $\mtx{A}$.

\begin{theorem}\label{Dudley:chaos}
Let $\vct{\epsilon}=(\epsilon_1,\ldots,\epsilon_n)^T$ be a Rademacher or Steinhaus sequence,
and let
\[
Y_{\vct{x}} := \vct{\epsilon}^\ast \mtx{B}(\vct{x}) \vct{\epsilon}= \sum_{q',q=1}^n  \overline{\epsilon_{q'}} \epsilon_q B(\vct{x})_{q',q}
\]
be an associated chaos process of order $2$, indexed by $x \in T$, where we
additionally assume $\mtx{B}(\vct{x})$ hermitian with zero diagonal, that is,
$B(\vct{x})_{q,q} = 0$ and $B(\vct{x})_{q',q} = \overline{B(\vct{x})_{q,q'}}$.
We define two (pseudo-)metrics on $T$,
	\begin{align*}
	d_1(\vct{x},\vct{y}) &=
		\|\mtx{B}(\vct{x})-\mtx{B}(\vct{y})\|_{2 \to 2},\notag\\
		d_2(\vct{x},\vct{y}) &= \|\mtx{B}(\vct{x}) - \mtx{B}(\vct{y})\|_F. 
	\end{align*}
	 Let $N(T,d_i,u)$ be the minimum number of balls of radius $u$ in the metric $d_i$
	 needed to cover $T$.  Then there exists a universal constant $K>0$ such that, for an arbitrary $\vct{x_0} \in T$,
	\begin{equation}
		\label{eq:dudley-chaos}
		\E\sup_{\vct{x}\in T} |Y_{\vct{x}} - Y_{\vct{x_0}}| ~\leq~ K\max\Big\{  \int_0^\infty \log N(T,d_1,u)~du
		\int_0^\infty \sqrt{\log N(T,d_2,u)}~du, \
	    \Big\}.
	\end{equation}
\end{theorem}

\begin{proof} For a Rademacher sequence, the theorem is stated in \cite[Proposition 2.2]{rarotr10}.
If $\vct{\epsilon}$ is a Steinhaus sequence and $\mtx{B}$ a Hermitian matrix then
\begin{align}
\vct{\epsilon}^* \mtx{B} \vct{\epsilon}  = \Re(\vct{\epsilon}^* \mtx{B} \vct{\epsilon}) & = \Re(\vct{\epsilon})^* \Re(\mtx{B}) \Re(\mtx{\epsilon})
- \Re(\vct{\epsilon})^* \Im(\mtx{B}) \Im(\vct{\epsilon}) \notag\\
&\quad + \Im(\vct{\epsilon})^* \Im(\mtx{B}) \Re(\vct{\epsilon}) + \Im(\vct{\epsilon})^* \Re(\mtx{B}) \Im(\vct{\epsilon}).  \notag
\end{align}
By decoupling, see,  for example,  \cite[Theorem 3.1.1]{gide99}, we have with $\epsilon'$ denoting an independent copy of
$\epsilon$,
\begin{align}
&\E \sup_{x \in T} |\Re(\vct{\epsilon})^* \Im (\mtx{B}(\vct{x})) \Im(\vct{\epsilon})| \leq 8 \, \E \sup_{x \in T} | \Re(\vct{\epsilon})^* \Im(\mtx{B}(\vct{x})) \Im(\vct{\epsilon}') |\notag\\
&\leq 8 \,\E \sup_{x \in T} | \vct{\xi}^* \Im(\mtx{B}(\vct{x})) \Im(\vct{\epsilon}')| \leq 8 \, \E \sup_{x \in T}  |\vct{\xi}^* \Im(\mtx{B}(\vct{x})) \vct{\xi}'|,\notag
\end{align}
where $\vct{\xi},\vct{\xi}'$ denote independent Rademacher sequences. The second and third inequalities follow from the
contraction principle \cite[Theorem 4.4]{leta91} (and symmetry of $\Re(\epsilon_\ell), \Im(\epsilon_\ell)$\, ) first applied conditionally on $\vct{\epsilon}'$ and then conditionally
on $\vct{\xi}$ (note that $|\Re(\epsilon_\ell)| \leq 1$, $|\Im(\epsilon_\ell)|\leq 1$ for all realizations of $\epsilon_\ell$). Using the triangle inequality we get
\begin{align}
\E \sup_{x \in T} | Y_{\vct{x}} - Y_{\vct{x_0}}| & \leq 16\, \E \sup_{x \in T} |\vct{\xi}^* (\Re(\mtx{B}(\vct{x})) - \Re(\mtx{B}(x_0)) \xi'|
\notag\\
&+ 16\, \E \sup_{x \in T} |\vct{\xi}^* (\Im(\mtx{B}(\vct{x})) - \Im(\mtx{B}(x_0))) \vct{\xi}'|.
\end{align}
Further note that $\|\Im(\mtx{B}(\vct{x})) - \Im(\mtx{B}(\vct{y})) \|_F,\ \|\Re(\mtx{B}(\vct{x})) - \Re(\mtx{B}(\vct{y})) \|_F \leq \|\mtx{B}(\vct{x}) - \mtx{B}(\vct{y})\|_F$ and similarly, writing $\mtx{B}(\vct{x})-\mtx{B}(\vct{y})$ as a $ 2n{\times}2n$ real block matrix acting on $\R^{2n}$ we see that also
$\|\Im(\mtx{B}(\vct{x})) - \Im(\mtx{B}(\vct{y})) \|_{2\to 2},\ \|\Re(\mtx{B}(\vct{x})) - \Re(\mtx{B}(\vct{y})) \|_{2\to 2} \leq \|\mtx{B}(\vct{x}) - \mtx{B}(\vct{y})\|_{2\to 2}$. Furthermore, the statement for Rade\-macher chaos processes holds
as well for decoupled chaos processes of the form above.
(Indeed, its proof uses decoupling in a crucial way.)
Therefore, the claim for Steinhaus sequences follows. 
\end{proof}

Note that $\mtx{B}(\vct{x})$ defined in \eqref{eq:def-B(x)} satisfies the hypotheses of
Theorem~\ref{Dudley:chaos} by definition. The pseudo-metrics are given by
\begin{align}
d_2(\vct{x},\vct{y}) \,=\, \| \mtx{B}(\vct{x}) - \mtx{B}(\vct{y}) \|_F =
     \Big( \sum_{q'\neq q}\big|\vct{x}^*\mtx{A}^\ast_{q'} \mtx{A}_q \vct{x} - \vct{y}^*\mtx{A}^\ast_{q'} \mtx{A}_q \vct{y}\big|^2\Big)^{1/2},
\end{align}
and
\[
d_1(\vct{x},\vct{y}) \,= \|\mtx{B}(\vct{x}) - \mtx{B}(\vct{y})\|_{2 \to 2}.
\]
The bound on the expected restricted isometry constant follows then from the following estimates on
the covering numbers of $T_s$ with respect to $d_1$ and $d_2$. Corresponding proofs will be detailed in Section
\ref{sec:cover:proofs}.
We start with $N(T_\sparsity,d_2,u)$.

\begin{lemma}\label{lem:d2:small} For $u > 0$, it holds
\[
\log(N(T_\sparsity,d_2,u)) \leq \sparsity \log(en^2/\sparsity)+ \sparsity \log(1+4 \sqrt{\sparsity n} u^{-1}).
\]
\end{lemma}
The above estimate is useful only for small $u > 0$. For large $u$ we require the following alternative bound.

\begin{lemma}\label{lem:d2:large} The diameter of $T_s$ with respect to $d_2$ is bounded by $4\sqrt{\sparsity n}$, and for
$\sqrt{n} \leq u \leq 4\sqrt{\sparsity n}$, it holds
\[
\log(N(T_\sparsity,d_2,u)) \leq  c u^{-2} n \sparsity^{3/2} \log(n \sparsity^{5/2} u^{-1}),
\]
where $c > 0$ is universal constant.
\end{lemma}

Covering number estimates with respect to $d_1$ are provided in the following lemma.

\begin{lemma}\label{lem:d1} The diameter of $T_s$ with respect to $d_1$ is bounded by $4s$, and for $u > 0$
\begin{align}
\log(N(T_\sparsity,d_1,u)) \leq
\min& \left\{ \sparsity \log(en^2/s) + \sparsity \log(1+4\sparsity u^{-1}), \right.\notag\\
& \left. \;\;\; c u^{-2} \sparsity^2 \log(2n) \log(n^2/u)\right\},\label{eqn:d1:estimate}
\end{align}
where $c > 0$ is a universal constant.
\end{lemma}

Based on these estimates and Theorem \ref{Dudley:chaos}
we  complete the proof of Theorem \ref{theorem:main}(a). By Lemmas \ref{lem:d2:small} and \ref{lem:d2:large},
the subgaussian integral in \eqref{eq:dudley-chaos} can be estimated as
\begin{align}
&\int_0^\infty \sqrt{\log(N(T_\sparsity,d_2,u))} du =
\int_0^{4 \sqrt{\sparsity n}}  \sqrt{\log(N(T_\sparsity,d_2,u))} du\notag\\
& = \int_0^{\sqrt{n}}  \sqrt{\log(N(T_s,d_2,u))} du +  \int_{\sqrt{n}}^{\sqrt{\sparsity n}} \sqrt{\log(N(T_s,d_2,u))} du\notag\\
& \leq \int_0^{\sqrt{n}} \sqrt{\sparsity \log(en^2/\sparsity)} du + \int_0^{\sqrt{n}} \sqrt{\sparsity \log(1+4 \sqrt{\sparsity n} u^{-1})} du
\notag\\
&+ c \sqrt{n \sparsity^{3/2}} \int_{\sqrt{n}}^{4\sqrt{\sparsity n}} u^{-1} \sqrt{\log(n \sparsity^{5/2} u^{-1}) }du \notag\\
& \leq \sqrt{\sparsity n \log( en^2/\sparsity)} + 4 \sparsity \sqrt{n} \int_0^{\sparsity^{-1/2}} \sqrt{\log(1+u^{-1})} du\notag\\
&+ c \sqrt{\sparsity^{3/2}n} \sqrt{\log(n^{1/2} \sparsity^{5/2})} \log(\sqrt{s})
\notag\\
& \leq  \sqrt{\sparsity n \log( en^2/\sparsity)} + 4 \sqrt{\sparsity n} \sqrt{\log(e(1+\sqrt{\sparsity}))} + c' \sqrt{\sparsity^{3/2} n \log(n) \log^2(\sparsity)}\notag\\
& \leq \hat{C}_1 \sqrt{\sparsity^{3/2} n \log(n) \log^2(\sparsity)}.\label{entropy:int1a}
\end{align}
Hereby, we have used \cite[Lemma 10.3]{ra10}, and that $\sparsity \leq n$.
Due to Lemma \ref{lem:d1} the subexponential integral obeys the estimate, for some $\kappa > 0$ to be chosen below,
\begin{align}
&\int_0^\infty \log(N(T_\sparsity,d_1,u)) du = \int_0^{4\sparsity} \log(N(T_\sparsity,d_1,u)) du\notag\\
&= \int_0^{\kappa} \log(N(T_\sparsity,d_1,u)) du + \int_{\kappa}^{4\sparsity} \log(N(T_\sparsity,d_1,u)) du\notag\\
& \leq \kappa \sparsity \log(en^2/\sparsity) + \sparsity \int_0^\kappa \log(1+4\sparsity u^{-1}) du
+ c \sparsity^2 \log(2n) \int_{\kappa}^{4\sparsity} u^{-2} \log(n^2/u) du\notag\\
& \leq \kappa \sparsity \log(en^2/\sparsity) + 4 \kappa \sparsity  \log(e(1+\kappa(4\sparsity)^{-1}))
+  c \sparsity^2 \kappa^{-1} \log(2n) \log(n^2/\kappa).\notag
\end{align}
Choose $\kappa = \sqrt{\sparsity \log(n)}$ to reach
\begin{align}
\int_0^\infty \log(N(T_\sparsity,d_1,u)) du \leq \hat{C}_2 \sparsity^{3/2} \log^{3/2}(n).\label{eqn:subexponential}
\end{align}
Combining the above integral estimates with \eqref{eqn:expect:delta} and Theorem \ref{Dudley:chaos} yields
\begin{equation}
\E \delta_s = \frac{1}{n} \E \sup_{x \in T_s} |Y_{\vct{x}} - Y_0 | \leq \frac{1}{n} \max\left\{C_1  \sqrt{\sparsity^{3/2} n \log(n) \log^2(\sparsity)} ,
C_2 \sparsity^{3/2} \log^{3/2}(n) \right\}.\label{Edelta:estimate}
\end{equation}
This is the statement of Theorem \ref{theorem:main}(a).

\begin{remark} \rm In analogy to the estimate of a subgaussian entropy integral arising in the analysis of partial random circulant
matrices in \cite{rarotr10}, we expect that the exponent $3/2$ in \eqref{entropy:int1a} can be improved to $1$. However, we doubt that
for the subexponential integral \eqref{eqn:subexponential} such improvement will be possible
(indeed, the estimate of the subexponential integral in \cite{rarotr10} also exhibits an exponent of $3/2$ at the $\sparsity$-term),
so that we did not pursue an improvement of \eqref{entropy:int1a} here as this would not provide a significant
overall improvement of \eqref{Edelta:estimate}. We expect that an improvement of \eqref{Edelta:estimate} would require more
sophisticated tools than the Dudley type estimate for chaos processes of Theorem \ref{Dudley:chaos}.
\end{remark}

\section{Proof of covering number estimates}
\label{sec:cover:proofs}

In this section we provide the covering number estimates of Lemma \ref{lem:d2:small}, \ref{lem:d2:large} and \ref{lem:d1},
which are crucial to the proof of our main result. We first introduce additional notation.
Let $\delta(m,k) = \delta_{0,m-k}$ and $\delta(m) = \delta_{0,m}$ be the Kronecker symbol as usual. We denote by $\supp \vct{x} = \{\ell, x_\ell \neq 0\}$ the support of a vector $\vct{x}$.
Let $\mtx{A}$ be a matrix with vector of singular values $\vct{\sigma}(\mtx{A})$. For $0<q\leq \infty$, the Schatten $S_q$-norm is defined by
\begin{align}\label{def:Sp}
\|\mtx{A}\|_{S_q} := \|\vct{\sigma}(\mtx{A})\|_q,
\end{align}
where $\|\cdot\|_q$ is the usual vector $\ell_q$ norm.
For an integer $p$, the $S_{2p}$ norm can be expressed as
\begin{eqnarray} \|\mtx{A}\|_{S_{2p}} = (\Tr((\mtx{A}^*
\mtx{A})^p))^{1/(2p)}. \label{eqn:Schatten2mNorm}
\end{eqnarray}
The $S_\infty$-norm coincides with the operator norm, $\|\cdot\|_{S_\infty} = \|\cdot\|_{2 \to 2}$.
By the corresponding properties of $\ell_q$-norms we have the inequalities
\begin{equation}\label{eqn:Schatten:ineq}
\|\mtx{A}\|_{2 \to 2} \leq \|\mtx{A}\|_{S_q} \leq \operatorname{rank}(\mtx{A})^{1/q} \|\mtx{A}\|_{2 \to 2}.
\end{equation}
Moreover, we will require an extension of the quadratic form $\mtx{B}(\vct{x})$ in \eqref{eq:def-B(x)} to a bilinear form,
\begin{equation}\label{def_Bx}
(\mtx{B}(\vct{x},\vct{z}))_{q',q} = \Big\{ \begin{array}{ll} \vct{x}^*\mtx{A}^\ast_{q'} \mtx{A}_q \vct{z} & \mbox{ if } q' \neq q,\\
0 & \mbox{ if } q' = q.
\end{array} \Big.
\end{equation}
Then $\mtx{B}(\vct{x})=\mtx{B}(\vct{x},\vct{x})$.

\subsection{Time--frequency analysis on $\C^n$}\label{section:t-f-analysis}

Before passing to the actual covering number estimates we provide some facts and estimates related to
time-frequency analysis on $\C^n$. Observe that the matrices $\mtx{A}_q$ introduced in
\eqref{eq:def-Aq} satisfy
$$
 \mtx{A}_q^\ast          = \left(
                       \begin{array}{c}
                         (\mtx{T}^q)^\ast
                         \\ (\mtx{M}\mtx{T}^q)^\ast
                         \\ (\mtx{M}^2\mtx{T}^q)^\ast \\ \vdots \\ (\mtx{M}^{n-1}\mtx{T}^q)^\ast
                       \end{array}
                     \right)
                     = \left(
                       \begin{array}{c}
                         \mtx{T}^{-q}
                         \\ \mtx{T}^{-q}\mtx{M}^{-1}
                         \\ \mtx{T}^{-q}\mtx{M}^{-2} \\ \vdots \\ \mtx{T}^{-q} \mtx{M}^{1}
                       \end{array}
                     \right),
 $$
 and, hence, $$(\mtx{A}_q^\ast \vct{y} )_{(k,\ell)}=y_{k+q}\ \omega^{-\ell(k+q)}.$$
 Clearly,
  \begin{eqnarray}
        \langle \mtx{A}_{q}\vct{z}, \vct{y} \rangle & = \langle \vct{z},  \mtx{A}_{q}^\ast \vct{y} \rangle
        = \sum_{k,\ell}  z_{(k,\ell)} \overline{y}_{k+q}  \omega^{\ell (k+q)}= \sum_{k,\ell}  z_{(k-q,\ell)} \overline{y}_{k}
         \omega^{\ell k}\notag\\
         & = \sum_{k} \big(\sum_{\ell} z_{(k-q,\ell)}  \omega^{\ell k} \big)\overline{y}_{k}\notag
        \end{eqnarray}
and, hence,
$$
(\mtx{A}_q \vct{z})_k = \sum_{\ell} z_{(k-q,\ell)}  \omega^{\ell k}.
$$
In the following, $\mtx{\mathcal F}:\C^n\mapsto \C^n$ denotes the normalized Fourier transform, that is,
\[
 (\mtx{\mathcal F} \vct{v})_\ell= n^{-1/2}\sum_{q=0}^{n-1} \omega ^{-q\ell} v_q.
\]
For $\vct{v}\in\C^{n\times n}$, $\mtx{\mathcal F}_2 \vct{v}$ denotes the Fourier transform in the second variable of $ v$.

Let $\{\vct{e}_\lambda\}_{\lambda\in \Z_n{\times
\Z_n}}$ and $\{\vct{e}_q\}_{q\in \Z_n}$ denoting the Euclidean basis of
$\C^{n\times n}$ respectively   $\C^n$, and, let $\mtx{P}_{\lambda}$ denote the
orthogonal projection onto the one dimensional space ${\rm span}\, \{ \vct{e}_{\lambda}\}$.
The following bounds will be crucial for the covering number estimates below.
\begin{lemma}\label{lemma:Gabor}
Let $\mtx{A}_q$ be as given in \eqref{eq:def-Aq}.  Then, for $\lambda \in \Z_n{\times
\Z_n}$, $q\in \Z_n$,
\begin{align} \label{eqn:AqPi}
  \mtx{A}_q \vct{e}_{\lambda} &= \mtx{\pi}(\lambda) \vct{e}_q\,,\\
  \sum_{q=0}^{n-1}\mtx{A}_q^\ast \mtx{A}_q & = n\, \mtx{I}\, \label{equation:sumWqqIsId},\\
  \sum_{q=0}^{n-1}\mtx{A}_q \mtx{P}_{\lambda} \mtx{A}_q^\ast & =  \mtx{I}\, \label{equation:sumWqqIsId2}\,, \\
   \sum_{q=0}^{n-1} \sum_{q'=0}^{n-1} \big| \vct{x}^\ast \mtx{A}_{q'}^\ast \mtx{A}_q \vct{y}  \big|^2 & \leq n\, \| \vct{x}\|_0  \, \|\vct{x}\|_2^2\,
    \| \vct{y}  \|_2^2 \label{equation:boundUsingSparsity}.
 \end{align}
\end{lemma}

\begin{proof}
For \eqref{eqn:AqPi}, observe that
\begin{align*}
  (\mtx{A}_q \vct{e}_{(k_0,\ell_0)})_k &=\sum_\ell \delta(k-q-k_0, \ell-\ell_0)\omega^{\ell k}
  =  \delta(q-(k-k_0))\omega^{\ell_0 k}\\
  &= (\mtx{\pi}(k_0,\ell_0)\vct{e}_q)_k\,.
\end{align*}
To see \eqref{equation:sumWqqIsId}, choose $\vct{z}\in \C^{n{\times}n}$ and
compute
 \begin{align*}
    \big(\mtx{A}_{q'}^\ast \mtx{A}_q \vct{z} \big)_{(k',\ell')}&=\sum_{\ell} z_{(k'+q'-q,\ell)} \omega^{\ell (k'+q')}\omega^{-\ell' (k'+q')}\\
    &=\sum_{\ell} z_{(k'+q'-q,\ell)} \omega^{(\ell-\ell') (k'+q')}\,.
   \end{align*}
 Hence,
 \begin{align*}
 \sum_q \big(\mtx{A}_{q}^\ast \mtx{A}_q \vct{z} \big)_{(k',\ell')}
        &=\sum_q\sum_{\ell} z_{(k',\ell)} \omega^{(\ell-\ell') (k'+q)} 
  =\sum_{\ell} z_{(k',\ell)} \sum_q \omega^{(\ell-\ell') (k'+q)}\\
 & =\sum_{\ell} z_{(k',\ell)} n\, \delta(\ell-\ell')
        =n\, z_{(k',\ell')}\,.
 \end{align*}
 Finally, observe that all but one column of $\mtx{A}_q \mtx{P}_{\{(\ell_0,k_0)\}}$ are 0, the nonzero column being
 column $(\ell_0,k_0)$, and only its $(k_0+q)$th entry is nonzero, namely, it is $\omega^{\ell_0(k_0+q)}$. We have
 $$\mtx{A}_q \mtx{P}_{\{(\ell_0,k_0)\}} \mtx{A}_q^\ast =\mtx{A}_q \mtx{P}_{\{(\ell_0,k_0)\}}\mtx{P}_{\{(\ell_0,k_0)\}} \mtx{A}_q^\ast =\mtx{A}_q \mtx{P}_{\{(\ell_0,k_0)\}} (\mtx{A}_q
 \mtx{P}_{\{(\ell_0,k_0)\}})^\ast,$$ and hence, $\mtx{A}_q \mtx{P}_{\{(\ell_0,k_0)\}} \mtx{A}_q^\ast=\mtx{P}_{\{k_0+q\}}$ and $\sum_q \mtx{A}_q \mtx{P}_{\{(\ell_0,k_0)\}} \mtx{A}_q^\ast=\mtx{I}.$

Let $\vct{x}\in\C^{n\times n}$ and $\Lambda=\supp \vct{x}$, then
 \begin{align}
  & \sum_q \sum_{q'} \big| \vct{x}^\ast \mtx{A}_{q'}^\ast \mtx{A}_q \vct{y}  \big|^2 =
            \sum_q \sum_{q'}\big| \sum_{(k',\ell')\in \Lambda} x_{(k',\ell')}\overline{\big(  \mtx{A}_{q'}^\ast \mtx{A}_q \vct{y}\big)}_{k',\ell'}  \big|^2
\notag \\
&\leq  \|\vct{x}\|_2^2  \
    \sum_q \sum_{q'} \sum_{(k',\ell')\in \Lambda} \big|
       \big(  \mtx{A}_{q'}^\ast \mtx{A}_q \vct{y}\big)_{k',\ell'}  \big|^2
       \notag \\
&=  \|\vct{x}\|_2^2 \
    \sum_q \sum_{q'} \sum_{(k',\ell')\in \Lambda} \big|
        \omega^{-\ell'(k'+q')} \sum_{\ell}\omega^{\ell (k'+q')}y_{(k'-(q-q'),\ell)}  \big|^2
              \notag \\
&=  \|\vct{x}\|_2^2 \
    \sum_q \sum_{q'} \sum_{(k',\ell')\in \Lambda} \big|
       \sum_{\ell}\omega^{\ell (k'+q')}y_{(k'-(q-q'),\ell)} \big|^2
       \notag\\
      & = n\,  \|\vct{x}\|_2^2
    \sum_{(k',\ell')\in \Lambda} \sum_q \sum_{q'} \big|
     \big(\mtx{{\mathcal F}}_2 \vct{y}\big)_{(k'-(q-q'),k'+q')}  \big|^2
       \notag\\
       &= n\,  \|\vct{x}\|_2^2
    \sum_{(k',\ell')\in \Lambda} \big\|
     \mtx{{\mathcal F}}_2 \vct{y}  \big\|^2_2
     \notag
       = n\, |\Lambda|  \, \|\vct{x}\|_2^2\,
    \| \vct{y}  \|_2^2 = n\,\|\vct{x}\|_0 \|\vct{x}\|_2^2 \|\vct{y}\|_2^2\notag
 \end{align}
by unitarity of $\mtx{{\mathcal F}}_2$. 
\end{proof}

\subsection{Proof of Lemma \ref{lem:d2:small}}
For  $\vct{x},\vct{y} \in \C^{n^2}$,
\begin{align}
d_2(\vct{x},\vct{y}) \,\leq\, \Big(\sum_{q'\neq q} \Big|\vct{x}^*\mtx{A}^\ast_{q'} \mtx{A}_q (\vct{x}-\vct{y})\Big|^2\Big)^{1/2}
+ \Big(\sum_{q'\neq q}\Big| (\vct{x}-\vct{y})^*\mtx{A}^\ast_{q'} \mtx{A}_q \vct{y}\Big|^2\Big)^{1/2}.\notag
\end{align}
Inequality \eqref{equation:boundUsingSparsity} implies that for $\vct{x},\vct{y} \in T_\sparsity$,
\begin{align}
\Big(\sum_{q'\neq q} \Big|\vct{x}^*\mtx{A}^\ast_{q'} \mtx{A}_q(\vct{x}-\vct{y})\Big|^2\Big)^{1/2}\Big(\sum_{q'\neq q}\Big| (\vct{x}-\vct{y})^*\mtx{A}^\ast_{q'} \mtx{A}_q \vct{y}\Big|^2\Big)^{1/2}
\leq \sqrt{\sparsity n} \, \|\vct{x}-\vct{y}\|_2 \notag
\end{align}
and, hence,
\begin{align}
d_2(\vct{x},\vct{y}) \,\leq\, 2\sqrt{\sparsity n}\, \|\vct{x}-\vct{y}\|_2.\label{d1:smallu}
\end{align}
Using the volumetric argument, see, for example,
\cite[Proposition 10.1]{ra10}, we obtain
\[
N(T_\sparsity,\|\cdot\|_2,u) \leq \Big(\begin{matrix} n^2 \\ \sparsity \end{matrix} \Big) (1+2/u)^\sparsity
\leq (en^2/\sparsity)^\sparsity(1+2/u)^\sparsity.
\]
By a rescaling argument
\begin{align*}
N(T_\sparsity,d_2,u) & \leq N(T_\sparsity, 2 \sqrt{\sparsity n} \|\cdot\|_2,u)
= N(T_\sparsity, \|\cdot\|_2, u/(2\sqrt{\sparsity n})) \\
&\leq (en^2/\sparsity)^\sparsity(1+4 \sqrt{\sparsity n}u^{-1})^\sparsity.
\end{align*}
Taking the logarithm completes the proof. 

\subsection{Proof of Lemma \ref{lem:d2:large}}
\label{subsection:d1largeU}

Now, we seek a suitable estimate of the covering numbers
$N(T_\sparsity,d_1,u)$ for
$u \geq \sqrt{n}$.
Observe that by \eqref{d1:smallu} the diameter of $T_\sparsity$ with respect
to $d_1$ is at most $4\sqrt{\sparsity n}$. Hence, it suffices to consider
$N(T_\sparsity,d_1,u)$ for
\begin{equation}\label{d1:u:range}
\sqrt{n} \leq u \leq 4\sqrt{\sparsity n},
\end{equation}
as stated in the lemma.
We use the empirical method \cite{ca85}, similarly as in \cite{ruve08}. We define the
norm $\|\cdot\|_\ast$ on $\C^{n{\times}n}$ by
\begin{eqnarray}
  \label{eqn:starnorm}\|\vct{x}\|_\ast= \sum_\lambda |\Re\, x_\lambda | +  |\Im \,
x_\lambda |\,.
\end{eqnarray}
For $\vct{x} \in T_\sparsity$ we define a random vector $\vct{Z}$, which takes 
$\|\vct{x}\|_\ast \sgn(\Re x_\lambda) \vct{e}_\lambda$ with probability $\frac{|\Re
\vct{x}_\lambda|}{\|\vct{x}\|_\ast}$, and the value  $i \|\vct{x}\|_\ast  \sgn(\Im x_\lambda)
\vct{e}_\lambda$ with probability $\frac{|\Im x_\lambda|}{\|\vct{x}\|_\ast}$.

Now, let $\vct{Z}_1,\ldots,\vct{Z}_m,\vct{Z}_1',\ldots,\vct{Z}_m'$ be independent copies of
$\vct{Z}$. We set $\vct{y}=\frac 1 m \sum_{j=1}^m \vct{Z}_j$ and $\vct{y}'=\frac 1 m \sum_{j=1}^m \vct{Z}_j'$ and
attempt to approximate $\mtx{B}(\vct{x})$ by
\begin{equation}\label{def:B}
\mtx{B} := \mtx{B}(\vct{y},\vct{y}')= \frac{1}{m^2} \sum_{j,j'=1}^m \mtx{B}(\vct{Z}_j,\vct{Z}_{j'}')\,.
\end{equation}
First, compute
\begin{align}
&\E \|\mtx{B} - \mtx{B}(\vct{x}) \|_{F}^2 = \E  \sum_{q,q'}\big|\vct{x}^* \mtx{W}_{q',q} \vct{x} - \frac{1}{m^2}\sum_{j,j'=1}^m \vct{Z}_j^* \mtx{W}_{q',q} \vct{Z}_{j'}'\big|^2\notag\\
& = \sum_{q,q'}\Big( |\vct{x}^* \mtx{W}_{q',q} \vct{x}|^2 - 2 \Re\Big(\overline{\vct{x}^* \mtx{W}_{q',q} \vct{x}} \, \E \Big[\frac{1}{m^2} \sum_{j,j'=1}^m \vct{Z}_j^* \mtx{W}_{q,q'} \vct{Z}_{j'}'\Big]\Big)\notag\\
  & \;\;\; + \E \Big[\Big|\frac{1}{m^2} \sum_{j,j'=1}^m \vct{Z}_j^* \mtx{W}_{q,q'} \vct{Z}_{j'}'\Big|^2\Big]\Big)\notag\\
&=   \sum_{q,q'}\Big(- |\vct{x}^* \mtx{W}_{q',q} \vct{x}|^2
     + \frac{1}{m^4} \sum_{j,j',j'',j'''=1}^m \E\Big[\vct{Z}_j^* \mtx{W}_{q,q'} \vct{Z}_{j'}' (\vct{Z}_{j''}')^* \mtx{W}_{q,q'}^* \vct{Z}_{j'''}\Big]\Big)\,,\notag
\end{align}
where we used that $\E[ \vct{Z}_j^* \mtx{W}_{q,q'} \vct{Z}_{j'}'] = \vct{x}^* \mtx{W}_{q,q'} \vct{x}$,
$j,j'=1,\ldots m$, by independence.
Moreover, for $j \neq j'''$ and $j'\neq j''$, independence implies
\[
 \E\Big[\vct{Z}_j^* \mtx{W}_{q,q'} \vct{Z}_{j'}' (\vct{Z}_{j''}')^* \mtx{W}_{q,q'}^* \vct{Z}_{j'''}\Big] =
  |\vct{x}^*\mtx{W}_{q,q'} \vct{x}|^2.
\]
To estimate  summands with  $ j'= j''$, note that
\[
\vct{Z}_{j}^* \mtx{W}_{q',q} \vct{Z}_{ j'}'(\vct{Z}_{j'}')^\ast \mtx{W}_{q,q'} \vct{Z}_{j'''}
= \|\vct{x}\|_\ast^2 \mtx{Z}_{ j}^* \mtx{A}^\ast_{q'} \mtx{A}_q \mtx{P}_{\{\lambda\}} \mtx{A}^\ast_{q}\mtx{A}_{q'} \vct{Z}_{j'''},
\]
where $\{\lambda\}=\supp \vct{Z}_{j'}$ is random.
Hence, in this case, we compute using \eqref{equation:sumWqqIsId2} in
Lemma~\ref{lemma:Gabor}
\begin{align}
&\sum_{q'\neq q} \E\Big[\vct{Z}_{ j}^* \mtx{A}^\ast_{q'} \mtx{A}_q \vct{Z}_{ j'}'(\vct{Z}_{ j'}')^\ast \mtx{A}^\ast_{q}\mtx{A}_{q'}  \vct{Z}_{j'''} \Big]\notag\\
& \leq  \|\vct{x}\|_\ast^2 \sum_{q', q} \E\Big[ \vct{Z}_{ j}^*\mtx{A}^\ast_{q'} \mtx{A}_q \mtx{P}_{\{\lambda\}} \mtx{A}^*_{q} \mtx{A}_{q'} \vct{Z}_{ j'''}\Big]\notag\\
&=  \|\vct{x}\|_\ast^2 \E \Big[ \vct{Z}_{ j}^*  \sum_{q'}
    \Big(\mtx{A}^\ast_{q'}\Big( \sum_q \mtx{A}_q \mtx{P}_{\{\lambda\}} \mtx{A}^*_{q} \Big) \mtx{A}_{q'}\Big) \vct{Z}_{ j'''}\Big]\notag\\
&=  \|\vct{x}\|_\ast^2 \E \Big[ \vct{Z}_{ j}^*  \sum_{q'}\Big(\mtx{A}^\ast_{q'} \mtx{A}_{q'}\Big) \vct{Z}_{ j'''}\Big]
= n \|\vct{x}\|_\ast^2 \E[\vct{Z}_{ j}^* \vct{Z}_{ j'''}]\notag\\
&=
\Big\{
  \begin{array}{ll}
    n \|\vct{x}\|_\ast^4, & \hbox{ if } j=j''', \\
    n \|\vct{x}\|_\ast^2 \E[\vct{Z}_{ j}^* ]\E[\vct{Z}_{ j'''}]=  n \|\vct{x}\|_\ast^2\|\vct{x}\|_2^2 \leq n \|\vct{x}\|_\ast^2, & \hbox{ else.}
  \end{array}
\Big.
 \notag
\end{align}
Symmetry implies an identical estimate for $j=j'''$, $j'\neq j''$.
As $\vct{x} \in T_\sparsity$ is $\sparsity$-sparse we have $\|\vct{x}\|_\ast \leq
\sqrt{2}\|\vct{x}\|_1\leq   \sqrt{2 \sparsity} \|\vct{x}\|_2 \leq \sqrt{2\sparsity}$. We
conclude
\begin{align}
  & \sum_{q',q}  \sum_{j,j',j'',j'''=1}^m \E\Big[\vct{Z}_j^* \mtx{W}_{q,q'} \vct{Z}_{j'}' (\vct{Z}_{j''}')^*
\mtx{W}_{q,q'}^* \vct{Z}_{j'''}\Big]\notag \\
&\leq m^2(m-1)^2  \sum_{q',q} |\vct{x}^\ast \mtx{W}_{q,q'} \vct{x}|^2 + m^2 n 4 s^2 + 2 m^2 (m-1) n \cdot 2s. \notag
\end{align}
For $m\geq \frac {11 n s^{\frac 3 2 }}{u^2}$ and $u\leq 4\sqrt{sn}$, we
finally obtain,
\begin{align}
&\E \|\mtx{B} - \mtx{B}(\vct{x}) \|_{F}^2  \leq  \sum_{q',q}
- |\vct{x}^* \mtx{W}_{q',q} \vct{x}|^2  + \frac{m^2(m^2-1)}{m^4}\sum_{q',q}  |\vct{x}^\ast \mtx{W}_{q,q'} \vct{x}|^2 \notag\\
&\;\;\;+   \frac {m^2 n 4 s^2} {m^4 } +  \frac {4m^2 (m-1) n s}{m^4}\notag  \\
&\leq   \frac {4  n s^2} {m^2 }+  \frac {4 n s}{m}
\leq   \frac {4  n s^2} {121 n^2s^3 } u^4+  \frac {4 n s}{11 n s^{\frac 3 2 } }u^2 \leq \frac { 64 n s} {121 n s } u^2+  \frac {44}{121 \sqrt{s} }u^2 \leq u^2\,.\label{Bx}
\end{align}

Since $\|\vct{x}\|_\ast$ can take any value in $[1, \sqrt{2\sparsity}]$,  we still have to
discretize this factor in the definition of the random variable $\vct{Z}$.
To this end, set
\[
\mtx{B}_\alpha := \frac{1}{m^2} \sum_{j=1,j'=1}^m \mtx{B}(\alpha \sgn(x_{\lambda_j}) \vct{e}_{\lambda_j}, \alpha
 \sgn(x_{\lambda'_{j'}}) \vct{e}_{\lambda'_{j'}})\,.
\]
Next, we observe that, for $\lambda=(k,\ell)$ and $\lambda'=(k',\ell')$,
\begin{eqnarray}
\mtx{B}(\vct{e}_{\lambda'}, \vct{e}_{\lambda})_{q',q} &= (\mtx{A}_{q'}e_{\lambda'})^\ast \mtx{A}_q
\vct{e}_{\lambda}
    =\langle \mtx{\pi}(\lambda)\vct{e}_q, \mtx{\pi}(\lambda') \vct{e}_{q'} \rangle \notag\\
    &= \Big\{
                                                               \begin{array}{ll}
                                                                \omega^{(\ell-\ell')(k+q)}, & \hbox{ if } k'+q'=k+q \,;\\
                                                                 0, & \hbox{ else,}
                                                               \end{array}
                                                             \Big. 
\label{eqn:B-elambda-elambdaprime}
\end{eqnarray}
and, hence,
$\|\mtx{B}(\vct{e}_{\lambda'}, \vct{e}_{\lambda})\|_F^2=n$.
Now, assume $\alpha$ is chosen such that $|\|\vct{x}\|_\ast^2 - \alpha^2| \leq \frac u
{\sqrt n}$. Then
\begin{align}
&\|\mtx{B}_\alpha - \mtx{B}_{\|\vct{x}\|_\ast} \|_F \notag\\
& = \Big\| \frac{1}{m^2} \sum_{j=1,j'=1}^m \mtx{B}(\alpha \sgn(x_{\lambda_j}) \vct{e}_{\lambda_j}, \alpha \sgn(x_{\lambda'_{j'}}) \vct{e}_{\lambda'_{j'}})
\notag\\
& \phantom{ \Big\| \frac{1}{m^2}} - \frac{1}{m^2} \sum_{j,j'=1}^m
\mtx{B}(\|\vct{x}\|_\ast \sgn(x_{\lambda_j}) \vct{e}_{\lambda_j}, \|\vct{x}\|_\ast \sgn(x_{\lambda'_{j'}}) \vct{e}_{\lambda'_{j'}})\Big\|_{F}\notag\\
&= |\|\vct{x}\|_\ast^2 - \alpha^2| \| \frac{1}{m^2} \sum_{j,j'=1}^m \mtx{B}(\sgn(x_{\lambda_j}) \vct{e}_{\lambda_j},\sgn(x_{\lambda'_{j'}})
\vct{e}_{k'_{j'}}) \|_{F}\notag\\
&  \leq \frac{u}{m^2\sqrt n} \sum_{j,j'=1}^{m} \|\mtx{B}(\vct{e}_{\lambda_j},\vct{e}_{\lambda_{j'}})\|_{F}\notag\\
& = u.
\label{Balpha}
\end{align}
We conclude that it suffices to choose $$K:=\Big\lceil\frac {2s -1} {\frac u
{\sqrt n} }\Big\rceil\leq \lceil 2s\sqrt n / u \rceil $$ values $\alpha_k \in J_s :=
[1,2s]$, $k=1,\ldots,K$, such that for each $\beta \in J_s$ there exists $k$
satisfying $| \beta -\alpha_k| \leq u / \sqrt{n}$.

Now, given $\vct{x}$ we can find $\vct{z}_1,\ldots,\vct{z}_m,\vct{z}_1',\ldots,\vct{z}_m'$ of the form
$\|\vct{x}\|_\ast p_\lambda \vct{e}_\lambda$, $p_\lambda \in\{1,-1,i,-i\}$ such that
$\|\mtx{B} -\mtx{B}(\vct{x})\|_F \leq u$. Further, we can find $k$ such that $|\|\vct{x}\|_\ast^2- \alpha_k^2|
\leq u/\sqrt{n}$. We replace the $\vct{z}_1,\ldots,\vct{z}_m,\vct{z}_1',\ldots,\vct{z}_m'$ by the
respective $\tilde{\vct{z}}_1,\ldots,\tilde{\vct{z}}_m,\tilde{\vct{z}}_1',\ldots,\tilde{\vct{z}}_m'$ of
the form $\alpha_j p_\lambda \vct{e}_\lambda$.

Then, using \eqref{Bx}, \eqref{Balpha} and the triangle inequality, we obtain
\[
\|\mtx{B}(\vct{x}) - \frac{1}{m^2} \sum_{j,j'=1}^m \mtx{B}(\tilde{\vct{z}}_j,\tilde{\vct{z}}'_{j'}) \|_{F} \leq 2u.
\]
Now, each $\tilde{\vct{z}_j}$, $\tilde{\vct{z}}_j'$ can take at most $\lceil 2s\sqrt n / u \rceil \cdot
4\cdot n^2$ values, so that $$\frac{1}{m^2} \sum_{j,j'=1}^m
\mtx{B}(\tilde{\vct{z}}_j,\tilde{\vct{z}}'_{j'})$$ can take at most $(4\lceil \frac{2s\sqrt{n}}{u} \rceil
n^2)^{2m} \leq (C s n^{\frac 5 2}/u)^{2m}$ values. Hence, we found a
$2u$-covering of the set of matrices $\mtx{B}(\vct{x})$ with $\vct{x} \in T_\sparsity$ of
cardinality at most $(C  s n^{\frac 5 2} /u)^{2m}$. Unfortunately, the matrices
of the covering are not necessarily of the form $\mtx{B}(\vct{x})$. Nevertheless, we
may replace each relevant matrix. (Clearly, if for a matrix
$\frac{1}{m^2} \sum_{j,j'=1}^m \mtx{B}(\tilde{\vct{z}}_j,\tilde{\vct{z}}_{j'}')$ there is no such
$\tilde{\vct{x}}$, then we can discard that matrix.) $\frac{1}{m^2} \sum_{j,j'=1}^m
\mtx{B}(\tilde{\vct{z}}_j,\tilde{\vct{z}}_{j'}')$ by a matrix $\mtx{B}(\tilde{\vct{x}})$ with
\[
\|\mtx{B}(\tilde{\vct{x}}) - \frac{1}{m^2} \sum_{j,j'=1}^m \mtx{B}(\tilde{\vct{z}}_j,\tilde{\vct{z}}'_{j'}) \|_{F} \leq 2u.
\]
Again, the set of such chosen $\tilde{\vct{x}}$ has cardinality at most $(C  s
n^{\frac 5 2} /u)^{2m}$ and, by the triangle inequality, for each $\vct{x}$ we can
find $\tilde{\vct{x}}$ of the covering such that
\[
d_2(\vct{x},\tilde{\vct{x}}) \leq 4u.
\]
For $m\geq 11 u^{-2} n s^{\frac 3 2}$, we consequently get
\[
\log(N(T_\sparsity,d_2,4u)) \leq
\log((C  s n^{\frac 5 2} /u)^{2m}) =
2 m \log( C n s^{5/2}/ u).
\]
The choice $m=\lceil 11 u^{-2} n s^{\frac 3 2}\rceil \leq 27 u^{-2} n s^{\frac 3
2}$ 
and rescaling gives
\[
\log(N(T_\sparsity,d_2,u))  \leq 27  u^{-2} n s^{\frac 3 2}  \log( 4C n s^{5/2}/ u) \leq c  u^{-2} n s^{\frac 3 2} \log( n s^{5/2}/ u).
\]
The proof of Lemma \ref{lem:d2:large} is completed. 

\subsection{Proof of Lemma \ref{lem:d1}, Part I}

Now we show the estimate
\[
\log(N(T_\sparsity,d_1,u)) \leq \sparsity \log(en^2/s) + \sparsity \log(1+4\sparsity u^{-1}),
\]
which will establish one part of \eqref{eqn:d1:estimate}. Before doing so, we note that one can quickly obtain an estimate for
$N(T_\sparsity,d_1,u)$ for small $u$ using that the Frobenius norm dominates the operator norm, and, hence $d_1(\vct{x},\vct{y}) \leq d_2(\vct{x},\vct{y}) \leq 2\sqrt{sn} \|\vct{x}-\vct{y}\|_2$. In fact, this estimate would not deteriorate the estimate in Theorem~\ref{theorem:main}(a).
But in the proof of Theorem~\ref{theorem:main}(b),
the more involved estimate $d_1(\vct{x},\vct{y})  \leq 2s \|\vct{x}-\vct{y}\|_2$ developed below is useful.

Let us first rewrite $d_1$. Recall \eqref{eqn:AqPi} in
Lemma~\ref{lemma:Gabor}, namely,
 $\mtx{A}_q \vct{e}_\lambda=\mtx{\pi}(\lambda) \vct{e}_q$, and, with $\lambda=(k,\ell)$ and $\lambda'=(k',\ell')$,
we obtain
$$\mtx{\pi}(\lambda')^\ast
\mtx{\pi}(\lambda)=\omega^{k'(\ell -\ell'')} \mtx{\pi}(\lambda-\lambda')\equiv
\omega(\lambda,\lambda')\mtx{\pi}(\lambda-\lambda').$$

Writing now  $\vct{x} = \sum_{\lambda \in \Z_n\times \Z_n} x_\lambda
\vct{e}_\lambda$, the entries of the matrix $\mtx{B}(\vct{x})$ in (\ref{def_Bx}) for $q'\neq q$
are given by
\begin{align}
& \mtx{B}(\vct{x})_{q'q}  = \sum_{\lambda,\lambda'} x_\lambda \overline{x}_{\lambda'} \vct{e}_{\lambda'}^*\mtx{A}^\ast_{q'} \mtx{A}_q \vct{e}_{\lambda}
= \sum_{\lambda,\lambda'} x_\lambda \overline{x}_{\lambda'} \vct{e}_{q'}^*\mtx{\pi}(\lambda')^* \mtx{\pi}(\lambda) \vct{e}_q\notag\\
&= \sum_{\lambda,\lambda'} x_\lambda \overline{x}_{\lambda'} \omega(\lambda,\lambda')\ \vct{e}_{q'}^*\mtx{\pi}(\lambda-\lambda')\vct{e}_q =
\sum_{\lambda \neq \lambda'} x_\lambda \overline{x}_{\lambda'} \omega(\lambda,\lambda')\ \vct{e}_{q'}^*\pi(\lambda-\lambda')\vct{e}_q\notag\\
& = \vct{e}_{q'}^* \Big(\sum_{\lambda\neq \lambda'} x_\lambda \overline{x}_{\lambda'} \omega(\lambda,\lambda')\ \mtx{\pi}(\lambda-\lambda')\Big) \vct{e}_q.
\notag
\end{align}
We used for the fourth inequality that  $\vct{e}_{q'}^*\mtx{\pi}(\ell_0,k_0)\vct{e}_q = 0$ if $q' \neq q$ and $k_0=0$.
This shows that
\[
\mtx{B}(\vct{x}) = \sum_{\lambda \neq \lambda'} x_\lambda \overline{x}_{\lambda'} \omega(\lambda,\lambda')\ \mtx{\pi}(\lambda-\lambda').
\]
The estimate \eqref{eqn:Schatten:ineq} for the Schatten norms shows
\begin{align}
d_1^{2p}(\vct{x},\vct{y}) & = \| \sum_{\lambda \neq \lambda'} (x_\lambda \overline{x}_{\lambda'}- y_\lambda \overline{y}_{\lambda'}) \omega(\lambda,\lambda')\ \mtx{\pi}(\lambda-\lambda') \|_{2 \to 2}^{2p}\notag\\
&\leq \|\sum_{\lambda \neq \lambda'} (x_\lambda \overline{x}_{\lambda'} - y_\lambda \overline{y}_{\lambda'}) \omega(\lambda,\lambda')\ \mtx{\pi}(\lambda-\lambda') \|_{S_{2p}}^{2p}
\notag\\
&= \sum_{\lambda_1\neq \lambda_1', \lambda_2 \neq \lambda_2', \ldots, \lambda_{2p} \neq \lambda_{2p}'}
(x_{\lambda_1}  \overline{x}_{\lambda_1'} - y_{\lambda_1}  \overline{y}_{\lambda_1'}) \cdots (x_{\lambda_{2p}}  \overline{x}_{\lambda_{2p}'} - y_{\lambda_{2p}}  \overline{y}_{\lambda_{2p}'}) \times \notag\\
& \phantom{=} \times \omega(\lambda_1,\lambda_1') \cdots \omega(\lambda_{2p},\lambda_{2p}')\
  \Tr\Big(\mtx{\pi}(\lambda_1-\lambda_1')\cdots \mtx{\pi}(\lambda_{2p}-\lambda_{2p}') \Big).\notag
\end{align}
Setting $(\ell_0,k_0)=\lambda_1 - \lambda_1'+ \lambda_2 - \lambda_2' +
\cdots + \lambda_{2p}-\lambda_{2p}'$ we observe that the trace in the last
expression sums over zero entries if $k_0\neq 0$ and sums over  roots of
unity to zero if $\ell_0\neq 0$. We conclude that
\[
\Big|\Tr\Big(\mtx{\pi}(\lambda_1-\lambda_1')\cdots \mtx{\pi}(\lambda_{2p}-\lambda_{2p}')\Big)\Big|
\leq n\,  \delta_{0,\lambda_1 - \lambda_1'+ \lambda_2 - \lambda_2' + \cdots + \lambda_{2p}-\lambda_{2p}'}
\,. \]
Hence,
\begin{align}
& d_1(\vct{x},\vct{y})^{2p} \leq n \sum_{\lambda_1 \neq \lambda_1'}
\big| x_{\lambda_1} \overline{x}_{\lambda_1'} - y_{\lambda_1} \overline{y}_{\lambda_1'}\big|
\ \sum_{\lambda_2 \neq \lambda_2'}
\big| x_{\lambda_2} \overline{x}_{\lambda_2'} - y_{\lambda_2}y_{\lambda_2'}\big|  \cdots \notag\\
& \cdots \sum_{\lambda_{2p-1}\neq \lambda_{2p-1}'}
 \big| x_{\lambda_{2p-1}} \overline{x}_{\lambda_{2p-1}'} - y_{\lambda_{2p-1}}\overline{y}_{\lambda_{2p-1}'}\big|
\  \sum_{\lambda_{2p}} \big| x_{\lambda_{2p}} \overline{x}_{\lambda_1- \lambda_1' + \cdots
    + \lambda_{2p}} - y_{\lambda_{2p}} \overline{y}_{\lambda_1 - \lambda_1' + \cdots + \lambda_{4p}}\big| . \notag
\end{align}
Now observe that, setting $t = \lambda_1 - \lambda_1'+ \cdots +
\lambda_{2p-1}- \lambda_{2p-1}'$, and using the Cauchy-Schwarz
inequality
\begin{align}
&\sum_{\lambda} |x_{\lambda} \overline{x}_{t + \lambda} - y_{\lambda} \overline{y}_{t + \lambda} |
\leq \sum_{\lambda} |x_{\lambda}| |x_{t+\lambda} - y_{t+\lambda}| + \sum_{\lambda}
|x_{\lambda} - y_{\lambda}| |y_{\lambda + t}|\notag\\
& \leq \|\vct{x}\|_2 \|\vct{x}-\vct{y}\|_2 + \|\vct{x}-\vct{y}\|_2 \|\vct{y}\|_2 = (\|\vct{x}\|_2 + \|\vct{y}\|_2) \|\vct{x}-\vct{y}\|_2.\notag
\end{align}
We obtain similarly
\begin{eqnarray}
\sum_{\lambda, \lambda'} | x_{\lambda} \overline{x}_{\lambda'} - y_{\lambda}
\overline{y}_{\lambda'}|&=& \sum_{\lambda, \lambda'} |x_{\lambda}| \, | x_{\lambda'} -
y_{\lambda'}| + | y_{\lambda'}| \, | x_\lambda-y_\lambda |
\leq (\|\vct{x}\|_1 + \|\vct{y}\|_1) \|\vct{x}-\vct{y}\|_1. \notag
\end{eqnarray}
For $\vct{x},\vct{y}$ with $\supp \vct{x} = \supp \vct{y} =\Lambda$ for $|\Lambda| \leq s$ and
$\|\vct{x}\|_2 = \|\vct{y}\|_2 = 1$ we have $\|\vct{x}\|_1 \leq \sqrt{\sparsity} \|\vct{x}\|_2 =
\sqrt{\sparsity}$ (and similarly for $\vct{y}$) as well as $\|\vct{x}-\vct{y}\|_1 \leq
\sqrt{\sparsity} \|\vct{x}-\vct{y}\|_2$. Hence, $$(\|\vct{x}\|_1 + \|\vct{y}\|_1) \|\vct{x}-\vct{y}\|_1 \leq 2\sparsity
\|\vct{x}-\vct{y}\|_2.$$
This finally yields
\[
d_1(\vct{x},\vct{y})^{2p} \leq 2^{2p} ns^{2p-1} \|\vct{x}-\vct{y}\|_2^{2p}
\]
for such $\vct{x},\vct{y}$. As this holds for all $p \in \N$ we conclude that
\begin{align}\label{d2:smallu}
d_1(\vct{x},\vct{y}) \leq 2\sparsity \|\vct{x}-\vct{y}\|_2.
\end{align}
With the volumetric argument, see for example\ \cite[Proposition 10.1]{ra10}, we obtain the bound
\[
\log(N(T_\sparsity, \|\cdot\|_2, u)) \leq \sparsity \log(en^2/\sparsity) +  \sparsity\log(1+2/u).
\]
Rescaling yields
\begin{align}
\log(N(T_\sparsity, d_1, u)) & \leq \log(N(T_\sparsity,2\sparsity \|\cdot\|_2,u)) =
\log(N(T_\sparsity,\|\cdot\|_2,u/(2\sparsity)))\notag\\
& \leq \sparsity \log(en^2/\sparsity) +  \sparsity\log(1+4\sparsity u^{-1}),\notag
\end{align}
which is the claimed inequality. 

\subsection{Proof of Lemma \ref{lem:d1}, Part II}

Next we establish the remaining estimate of \eqref{eqn:d1:estimate},
\[
\log(N(T_\sparsity,d_1,u)) \leq c u^{-2} \sparsity^2 \log(2n) \log(n^2/u).
\]
To this end, we use again the empirical method as in
Section~\ref{subsection:d1largeU}.

For $\vct{x} \in T_\sparsity$, we define $\vct{Z}_1,\ldots, \vct{Z}_m$ and $\vct{Z}_1',\ldots,
\vct{Z}_m'$ as in Section~\ref{subsection:d1largeU}, that is, each takes
independently the value $\|\vct{x}\|_\ast \sgn(\Re x_\lambda) \vct{e}_\lambda$ with
probability $\frac{|\Re x_\lambda|}{\|\vct{x}\|_\ast}$, and the value
$i \|\vct{x}\|_\ast \sgn(\Im x_\lambda) \vct{e}_\lambda$ with probability $\frac{|\Im
x_\lambda|}{\|\vct{x}\|_\ast}$.

As before, we set
\begin{equation}\label{def:BZZ}
B(\vct{Z},\vct{Z}') =  (\vct{Z}^* \vct{W}_{q'q} \vct{Z}')_{q',q},
\end{equation}
where $\mtx{A}^\ast_{q'} \mtx{A}_q =\mtx{A}^\ast_{q'} \mtx{A}_q$ for $q' \neq q$ and $\mtx{W}_{q,q} =
0$, $j = 1,\ldots,N$, and attempt to approximate $\mtx{B}(\vct{x})$ with
\begin{equation}\label{def:B2}
\mtx{B} := \frac{1}{m} \sum_{j=1}^m \mtx{B}(\vct{Z}_j,\vct{Z}_j').
\end{equation}
That is, we will estimate $\E \|\mtx{B} - \mtx{B}(\vct{x})\|_{2 \to 2}^2$.

We will use symmetrization as formulated in the following
lemma \cite[Lemma 6.7]{ra10}, see also \cite[Lemma 6.3]{leta91}, \cite[Lemma 1.2.6]{gide99}.
Note that we will use this result with
$\mtx{B}_j=\mtx{B}(\vct{Z}_j,\vct{Z}_j')$.
\begin{lemma}\label{lem:symmetrize} (Symmetrization) Assume that $(\vct{Y}_j)_{j=1}^{m}$ is a
sequence of independent random vectors in $\C^r$ equipped with a
(semi-)norm $\|\cdot\|$, having expectations $\beta_j = \E \vct{Y}_j$. Then for
$1\leq p < \infty$
\begin{equation}\label{symmetrize}
\Big(\E \|\sum_{j=1}^m (\vct{Y}_j - \beta_j) \|^p\Big)^{1/p} \leq 2 \Big(\E \|\sum_{j=1}^m \epsilon_j \vct{Y}_j \|^p \Big)^{1/p},
\end{equation}
where $(\epsilon_j)_{j=1}^N$ is a Rademacher series independent of $(\vct{Y}_j)_{j=1}^{m}$.
\end{lemma}
To estimate the $2p$-th moment of $\| \mtx{B}(\vct{x}) - \mtx{B}\|_{2\to 2}$,  we will use the
noncommutative Khintchine inequality \cite{bu01,ra10} which makes use of
the Schatten $p$-norms introduced in \eqref{def:Sp}.
\begin{theorem}\label{thm:Khintchine:nonc} (Noncommutative Khintchine inequality) Let
$\vct{\epsilon}=(\epsilon_1,\ldots,\epsilon_m)$ be a Rademacher sequence, and
let $\mtx{A}_j$, $j=1,\ldots,m$, be complex matrices of the same dimension.
Choose $p \in \N$. Then
\begin{align}
&\E \| \sum_{j=1}^m \epsilon_j \mtx{A}_j \|_{S_{2p}}^{2p} \leq \frac{(2p)!}{2^p p!} \max \Big\{ \Big\|\Big(\sum_{j=1}^m \mtx{A}_j \mtx{A}_j^*\Big)^{1/2} \Big\|_{S_{2p}}^{2p}, \Big\|\Big(\sum_{j=1}^m \mtx{A}_j^* \mtx{A}_j\Big)^{1/2} \Big\|_{S_{2p}}^{2p} \Big\}.\label{ineq:Khintchine:nonc}
\end{align}
\end{theorem}
Let $p \in \N$. We apply symmetrization with $\mtx{B}_j=\mtx{B}(\vct{Z}_j,\vct{Z}_j')$, estimate the
operator norm by the Schatten-$2p$-norm and apply the noncommutative
Khintchine inequality (after using Fubini's theorem), to obtain
\begin{align}
& \Big( \E \|\mtx{B}-\mtx{B}(\vct{x})\|_{2 \to 2}^{2p}\Big)^{\frac{1}{2p}}
= \Big( \E \|\frac{1}{m} \sum_{j=1}^m (\mtx{B}(\vct{Z}_j,\vct{Z}_j') - \E \mtx{B}(\vct{Z}_j,\vct{Z}_j') )\|_{2 \to 2}^{2p}\Big)^{\frac{1}{2p}}
\notag\\
& \leq \frac{2}{m} \Big( \E \| \sum_{j=1}^m \epsilon_j \mtx{B}(\vct{Z}_j,\vct{Z}_j') \|_{2 \to 2}^{2p} \Big)^{\frac{1}{2p}}
\leq \frac{2}{m} \Big( \E \| \sum_{j=1}^m \epsilon_j \mtx{B}(\vct{Z}_j,\vct{Z}_j') \|_{S_{2p}}^{2p} \Big)^{\frac{1}{2p}}\notag\\
& \leq \frac{2}{m} \Big(\frac{(2p)!}{2^p p!}\Big)^{\frac{1}{2p}}
\Big( \E \max\Big\{ \Big\|\Big( \sum_{j=1}^m \mtx{B}(Z_j,Z_j')^* \mtx{B}(\vct{Z}_j,\vct{Z}_j') \Big)^{1/2} \Big\|_{S_{2p}}^{2p},\notag\\
& \phantom{ \frac{2}{m} \Big(\frac{(2p)!}{2^p p!}\Big)^{\frac{1}{2p}} \Big( \E \max}
\Big\|\Big( \sum_{j=1}^m \mtx{B}(\vct{Z}_j,\vct{Z}_j') \mtx{B}(\vct{Z}_j,\vct{Z}_j')^* \Big)^{1/2} \Big\|_{S_{2p}}^{2p}\Big\} \Big)^{\frac{1}{2p}}. \label{maxxx}
\end{align}
Now recall that the $\vct{Z}_j,\vct{Z}_j'$ may take the values $\|\vct{x}\|_\ast p_\lambda
\vct{e}_\lambda$, with\linebreak $p_\lambda\in \{ 1,-1,i, -i \}$. Further, observe that
$\mtx{B}(\vct{e}_{\lambda'},\vct{e}_\lambda)^* = \mtx{B}(\vct{e}_\lambda,\vct{e}_{\lambda'})$, and, for $q
\neq q'$,
\begin{align}
& (\mtx{B}(\vct{e}_{\lambda'},\vct{e}_\lambda)^*\mtx{B}(\vct{e}_{\lambda'},\vct{e}_\lambda))_{q,q''}
    = \sum_{q'} \vct{e}_{\lambda}^\ast \mtx{A}_{q}^\ast \mtx{A}_{q'} e_{\lambda'} \, \vct{e}_{\lambda'}^\ast \mtx{A}_{q'}^\ast \mtx{A}_{q''} \vct{e}_{\lambda}\notag\\
   &= \sum_{q'} \vct{e}_{\lambda}^\ast \mtx{A}_{q}^\ast \mtx{A}_{q'} \mtx{P}_{\lambda'} \mtx{A}_{q'}^\ast \mtx{A}_{q''} \vct{e}_{\lambda}
= \vct{e}_{\lambda}^\ast \mtx{A}_{q}^\ast  \big(\sum_{q'} \mtx{A}_{q'} \mtx{P}_{\lambda'} \mtx{A}_{q'}^\ast \big)\ \mtx{A}_{q''} \vct{e}_{\lambda}\notag\\
&= \vct{e}_{\lambda}^\ast \mtx{A}_{q}^\ast \mtx{A}_{q''} \vct{e}_{\lambda}
= \langle \mtx{\pi}(\lambda)\vct{e}_{q''},\mtx{\pi}(\lambda) \vct{e}_{q} \rangle =  \langle \vct{e}_{q''},\vct{e}_{q} \rangle =\delta(q''-q).\notag
\end{align}
Therefore, $\mtx{B}(\vct{e}_{\lambda'},\vct{e}_\lambda)^*\mtx{B}(\vct{e}_{\lambda'},\vct{e}_\lambda) = \mtx{I}$
and
\begin{align}\label{BstarB}
\mtx{B}(\vct{Z}_\ell,\vct{Z}_\ell')^*\mtx{B}(\vct{Z}_j,\vct{Z}_j') = \|\vct{x}\|_\ast^4 \mtx{I}.
\end{align}
Since $\|\mtx{I} \|_{S_{2p}}^{2p} =n $, $ \|\vct{x}\|_\ast \leq 2 s \|\vct{x}\|_2=2s$, we obtain
\begin{align}
&\|\Big( \sum_{j=1}^m \mtx{B}(\vct{Z}_j,\vct{Z}_j')^* \mtx{B}(\vct{Z}_j,\vct{Z}_j') \Big)^{1/2}\|^{2p}_{S_{2p}}
= \| \Big( \sum_{j=1}^m  \|\vct{x}\|_\ast ^4  \mtx{I}  \Big)^{1/2}  \|^{2p}_{S_{2p}}
= \|\vct{x}\|_\ast^{4p} m^p n\notag\\
&\leq (2s)^{2p} m^p n\,. \label{empirical:estimate}
\end{align}
By symmetry this inequality applies also to the second term in the maximum
in \eqref{maxxx}. This yields
\begin{align}
\Big( \E \|\mtx{B}-\mtx{B}(\vct{x})\|_{2 \to 2}^{2p}\Big)^{\frac{1}{2p}}
        \leq  \frac{2}{m} \Big(\frac{(2p)!}{2^q q!}\Big)^{\frac{1}{2p}} \ 2s m^{\frac 1 2} n^{\frac 1 {2p}}
\leq  \frac{4s}{\sqrt{m}}n^{1/(2p)} \Big(\frac{(2p)!}{2^p p!}\Big)^{\frac{1}{2p}}.\notag
\end{align}
Using H\"older's inequality, we can interpolate between $2p$ and $2p+2$,
and an application of Stirling's formula yields for arbitrary moments $p \geq
2$, see also \cite{ra10},
\begin{align}
\Big( \E \|\mtx{B}-\mtx{B}(\vct{x})\|_{2 \to 2}^{p}\Big)^{1/p}
\leq 2^{3/(4p)} n^{1/p} e^{-1/2} \sqrt{p} \frac{4s}{\sqrt{m}}.
\end{align}
Now we use the following lemma relating moments and tails
\cite{ra09,ra10}.
\begin{proposition}\label{prop:moments} Suppose $\Xi$ is a
random variable satisfying
\[
(\E |\Xi|^p)^{1/p} \leq \alpha \beta^{1/p} p^{1/\gamma}\quad \mbox{ for all } p \geq p_0
\]
for some constants $\alpha,\beta,\gamma, p_0 > 0$. Then
\[
\P(|\Xi| \geq e^{1/\gamma} \alpha v) \leq \beta e^{- v^\gamma / \gamma}
\]
for all $v \geq p_0^{1/\gamma}$. 
\end{proposition}

Applying the lemma with $p_0=2$, $\gamma=2$, $\beta=2^{3/4}n$,
$\alpha=e^{-1/2}
 \frac{4s}{\sqrt{m}}$, and
$$v=u \frac{e^{-1/\gamma}}{\alpha} =u \frac{e^{-1/2} \sqrt{m}}{ e^{-1/2} 4s} = u\frac{\sqrt m }{ 4s}\geq \sqrt 2 $$
gives
\[
\P\Big( \|\mtx{B} - \mtx{B}(\vct{x}) \|_{2 \to 2} \geq u \Big) \leq 2^{3/4} n e^{-\frac{m u^2}{32s^2}}, \quad u \geq 4s \sqrt{2/m}.
\]
In particular, if
\begin{align}\label{m:ineq2}
m > \frac{32s^2}{u^2} \log(2^{3/4}n)
\end{align}
then there exists a matrix  of the form $ \frac{1}{m} \sum_{j=1}^m \mtx{B}(\vct{z}_j,\vct{z}_j')$
with $\vct{z}_j,\vct{z}_j'$ of the given form $\|\vct{x}\|_\ast p_\lambda \vct{e}_\lambda $ for some
$k$ such that
\[
\Big\| \frac{1}{m} \sum_{j=1}^m \mtx{B}(\vct{z}_j,\vct{z}_j') - \mtx{B}(\vct{x}) \Big\| \leq u.
\]
As before,  we still have to discretize the prefactor $\|\vct{x}\|_\ast$. Assume that
$\alpha$ is chosen such that $|\|\vct{x}\|_\ast^2 - \alpha^2| \leq u$. Then, similarly
as in \eqref{Balpha},
\begin{align}
&\Big\| \frac{1}{m} \sum_{j=1}^m \mtx{B}(\alpha \sgn(x_{\lambda_j}) \vct{e}_{\lambda_j}, \alpha \sgn(x_{\lambda_{j'}}) \vct{e}_{\lambda_{j'}}) \notag\\
& \;\;\;
- \frac{1}{m} \sum_{j=1}^m \mtx{B}(\|\vct{x}\|_1 \sgn(x_{\lambda_j}) \vct{e}_{\lambda_j}, \|\vct{x}\|_1 \sgn(x_{\lambda_{j'}}) \vct{e}_{\lambda_{j'}})\Big\|_{2 \to 2}\notag\\
&= |\|\vct{x}\|_1^2 - \alpha^2| \| \frac{1}{m} \sum_{j=1}^m \mtx{B}(\sgn(x_{\lambda_j}) \vct{e}_{\lambda_j},\sgn(x_{\lambda_{j'}})
\vct{e}_{\lambda_{j'}}) \|_{2 \to 2}\notag\\
& \leq \frac{u}{m} \sum_{j=1}^m \|\mtx{B}(\sgn(x_{\lambda_j}) \vct{e}_{\lambda_j},\sgn(x_{\lambda_{j'}})
\vct{e}_{\lambda_{j'}})\|_{2 \to 2}
= u.\notag
\end{align}
Hereby, we used $\|\mtx{B}(\sgn(x_{\lambda_j})
\vct{e}_{\lambda_j},\sgn(x_{\lambda_{j'}}) \vct{e}_{\lambda_{j'}})\|_{2 \to 2} = 1$.

As in Section~\ref{subsection:d1largeU}, we use a discretization of $J_s=
[1,2s]$ with about $K=\lceil \frac{2s}{u} \rceil$ elements,
$\alpha_1,\ldots,\alpha_K$ such that for any $\beta$ in $J_s$ there exists
$k$ such $| \beta-\alpha_k^2 | \leq u$. Now, provided \eqref{m:ineq2} holds,
for given $\vct{x}$ we can find
$\tilde{\vct{z}}_1,\ldots,\tilde{\vct{z}}_m,\tilde{\vct{z}}_1',\ldots,\tilde{\vct{z}}_m'$ of the form
$\alpha_k  \sgn(x_\lambda) \vct{e}_\lambda$, $p(\lambda)\in\{1,-1,i,-i\}$, with
\[
\|\mtx{B}(\vct{x}) - \frac{1}{m} \sum_{j=1}^m \mtx{B}(\tilde{\vct{z}}_j,\tilde{\vct{z}}_j') \|_{2 \to 2} \leq 2u.
\]
Observe as in Section~\ref{subsection:d1largeU} that each $\tilde{\vct{z}_j}$ can
take $4\lceil \frac{2s}{u} \rceil n^2$ values, so that $\frac{1}{m} \sum_{j=1}^m
B(\tilde{z}_j,\tilde{z}_j')$ can take at most $(4\lceil \frac{2s}{u} \rceil n^2)^{2m}
\leq (C n^2 s/u)^{2m}$ values. As seen before, this establishes a $4u$
covering of the set of matrices $\mtx{B}(\vct{x})$ with $\vct{x} \in T_\sparsity$ of cardinality
at most $(C n^2 s/u)^{2m}$, and we conclude
\begin{align}
\log(N(T_\sparsity,d_1,u)) &\leq \log( (Cn^2s/u)^{2m}) \leq C' \frac{s^2}{u^2} \log(2^{3/4}n) \log(Cn^2s/u)\notag\\
&\leq \tilde{C} \frac{s^2}{u^2} \log(2n)  \log(n^2/u).\notag
\end{align}
This completes the proof of Lemma \ref{lem:d1}. 

\section{Probability estimate}\label{section:probability}

To prove Theorem~\ref{theorem:main}(b) will use the following concentration inequality, which is a slight variant of
Theorem 17 in \cite{boluma03}, which in turn is an improved version of a
striking result due to Talagrand \cite{ta96-2}. Note that with $\mtx{B}(\vct{x})$ as
defined above, $Y$ below satisfies $\E Y=n \,\E \delta_s$.
\begin{theorem}
\label{th:boucheron} Let $\mathscr{B}=\{\mtx{B}(\vct{x})\}_{\vct{x}\in T}$ be a countable
collection of
$n\times n$ complex Hermitian matrices, and let
$\vct{\epsilon}=(\epsilon_1,\ldots,\epsilon_n)^T$ be a sequence of i.i.d. Rademacher or
Steinhaus random variables.
Assume that $B(\vct{x})_{q,q} = 0$ for all $\vct{x}\in T$. Let $Y$ be the random
variable
\[
		Y     = \sup_{\vct{x}\in T} \Big|\vct{\epsilon}^\ast \mtx{B}(\vct{x}) \vct{\epsilon} \Big|
                = \Big| \sum_{q,q'=1}^n \overline{\epsilon_{q'}} \epsilon_q B(\vct{x})_{q',q} \Big|.
\]
Define $U$ and $V$ to be
\[
	U = \sup_{\vct{x}\in T} \|\mtx{B}(\vct{x})\|_{2 \to 2}
\]
and
\begin{align}\label{def:V}
	V = \E\sup_{\vct{x}\in T}\|\mtx{B}(\vct{x})\vct{\epsilon}\|_2^2= \E\sup_{\vct{x}\in T}\sum_{q'=1}^n\Big|\sum_{q=1}^n \epsilon_q B(\vct{x})_{q',q}\Big|^2.
\end{align}
Then, for $\lambda \geq 0$,
	\begin{equation}
		\label{eq:boucheron}
		\P\Big(Y \geq \E[Y] + \lambda\Big)
		~\leq~
		\exp\Big(-\frac{\lambda^2}{32V + 65U\lambda/3} \Big).
	\end{equation}
\end{theorem}
\begin{proof} For Rademacher variables, the statement is exactly Theorem 17 in \cite{boluma03}.
For Steinhaus sequences, we provide a variation of its proof.
For $\vct{\epsilon} = (\epsilon_1,\ldots,\epsilon_n)$,
let $g_{\mtx{M}}(\vct{\epsilon}) = \sum_{j,k=1}^n \overline{\epsilon_j} \epsilon_k M_{j,k}$
and set
\[
Y = f(\vct{\epsilon}) = \sup_{\mtx{M}\in\mathscr{B}} \Big| g_{\mtx{M}}(\vct{\epsilon}) \Big|.
\]
Further, for an independent copy $\widetilde{\epsilon}_\ell$ of $\epsilon_\ell$, set
$\vct{\epsilon^{(\ell)}} = (\epsilon_1,\ldots,\epsilon_\ell,\widetilde{\epsilon_\ell},\epsilon_{\ell+1},\ldots,\epsilon_n)$ and $Y^{(\ell)} = f(\vct{\epsilon^{(\ell)}})$.
Conditional on $(\epsilon_1,\ldots,\epsilon_n)$,
let $\widehat{\mtx{M}} = \widehat{\mtx{M}}(\vct{\epsilon})$ be the matrix giving the maximum in the definition of $Y$.
(If the supremum is not attained, then one has to consider finite subsets $T \subset \mathscr{B}$. The derived estimate
will not depend on $T$, so that one can afterwards pass over to the possibly infinite, but countable, set $\mathscr{B}$.)
Then we obtain, using $\widehat{\mtx{M}}^* = \widehat{\mtx{M}}$ and $\widehat{M}_{kk} = 0$
in the last step,
\begin{align}
&\E \Big[ (Y-Y^{(\ell)})^2 \mathbf{1}_{Z>Z^{(\ell)}} | \epsilon\Big]  \leq
\E \Big[ |g_{\widehat{M}}(\vct{\epsilon}) - g_{\widehat{M}}(\vct{\epsilon^{(\ell)}})|^2 \mathbf{1}_{Z>Z^{(\ell)}} | \epsilon\Big]
\notag\\
& = \E \Big[ |(\overline{\epsilon_\ell - \widetilde{\epsilon_\ell}}) \sum_{j=1, j\neq \ell}^n \epsilon_j \widehat{M}_{j,\ell} + (\epsilon_\ell - \widetilde{\epsilon_\ell}) \sum_{k=1,k\neq \ell}^n \overline{\epsilon_k} \widehat{M}_{\ell,k}|^2 \mathbf{1}_{Z>Z^{(\ell)}} | \epsilon\Big]\notag\\
&\leq 4 \E_{\widetilde{\epsilon_\ell}} |\epsilon_\ell - \widetilde{\epsilon_\ell}|^2
\Big|\sum_{j=1,j\neq \ell}^n \epsilon_j \widehat{M}_{j,\ell}\Big|^2
= 8  \Big|\sum_{j=1}^n\epsilon_j \widehat{M}_{j,\ell}\Big|^2.\notag
\end{align}
The remainder of the proof is analogous to the one in \cite{boluma03} and therefore omitted. 
\end{proof}

We first note that we may pass from $T_s$ to a dense countable subset
$T_s^{\circ}$ without changing the supremum, hence Theorem~\ref{th:boucheron} is applicable. 
Now, it remains to estimate $U$ and
$V$.  To this end, note that \eqref{d2:smallu} implies
$$
    U= \sup_{\vct{x}\in T_\sparsity} \|\mtx{B}(\vct{x})\|_{2 \to 2} \leq \sup_{\vct{x}\in T_\sparsity} 2s \|\vct{x}\|_2=2s\, .
$$

The remainder of this section develops an estimate of the quantity $V$ in
\eqref{def:V}. Hereby, we rely on a Dudley type inequality for Rademacher or Steinhaus
processes with values in $\ell_2$, see below. First we note the following
Hoeffding type inequality.

\begin{proposition}\label{prop:Hoeffding:vec} Let $\vct{\epsilon}=(\epsilon_q)_{q=1}^n$ be a
Steinhaus sequence and let $\mtx{B} \in \C^{m\times n}$. Then, for $u\geq 0$,
\begin{align}\label{PSteinhaus}
\P\Big(\| \mtx{B}\vct{\epsilon}\|_2 \geq u \|\mtx{B}\|_F \Big) \leq 8 e^{-u^2/16}.
\end{align}
\end{proposition}

\begin{proof}
In \cite[Proposition B.1]{rarotr10}, it is shown that
 \begin{equation}\label{Rademacher:Hoeffding}
\P\Big(\| \mtx{B}\vct{\epsilon} \|_2 \geq u \|\mtx{B}\|_F \Big) \leq 2 e^{-u^2/2}.
\end{equation}
for Rademacher sequences.  We extend this result using the contraction principle \cite[Theorem 4.4]{leta91}, as in the proof of Theorem \ref{Dudley:chaos}.

In fact, \cite[Theorem 4.4]{leta91} implies that for $\mtx{B}\in\C^{n\times n}$ and $\vct{\epsilon}$ being a Steinhaus sequence and
$\vct{\xi}$ a Rademacher sequence, we have, for example
\[
  \P(\|\Re (\mtx{B}) \Re (\vct{\epsilon}) \|_2 \geq u\|\mtx{B}\|_F)\leq 2 \P(\|\Re \mtx{B} \vct{\xi} \|_2 \geq u\|\mtx{B}\|_F)\leq 4 e^{-u^2/2}.
\]
Hence,
\begin{align*}
&  \P(\| \mtx{B}  \vct{\epsilon} \|_2 \geq u\|\mtx{B}\|_F) = \P(\| \Re(\mtx{B}  \vct{\epsilon}) \|_2^2
        +\| \Im(\mtx{B}  \vct{\epsilon}) \|_2^2 \geq u^2\|\mtx{B}\|_F^2)\\
  &\leq \P(\| \Re(\mtx{B}  \vct{\epsilon}) \|_2^2 \geq \frac{u^2}{\sqrt 2})
        +\P(\| \Im(\mtx{B}  \vct{\epsilon}) \|_2^2 \geq \frac{u}{\sqrt 2} \|\mtx{B}\|_F^2)\\
  &\leq \P(\| \Re \mtx{B} \Re \vct{\epsilon}) \|_2 \geq \frac{u}{\sqrt 8}\|\mtx{B}\|_F^2) +
        \P(\| \Im \mtx{B} \Im \vct{\epsilon}) \|_2 \geq \frac{u}{\sqrt 8}\|\mtx{B}\|_F^2) \\
&        +\P(\| \Re \mtx{B} \Im \vct{\epsilon}) \|_2 \geq \frac{u}{\sqrt 8}\|\mtx{B}\|_F^2) +
        \P(\| \Im \mtx{B} \Re \vct{\epsilon}) \|_2 \geq \frac{u}{\sqrt 8}\|\mtx{B}\|_F^2)\\
       & \leq 8 e^{-u^2/16}.
\end{align*}
\end{proof}
With more effort, one may also derive \eqref{PSteinhaus} with better constants.
Let us now estimate the quantity
\[
V   = \E \sup_{\vct{x} \in T_\sparsity} \|\mtx{B}(\vct{x})\vct{\epsilon}\|_2^2
        =\E \sup_{\vct{x} \in T_\sparsity} \sum_{q'=1} |\sum_{q=1} \epsilon_{q} B(\vct{x})_{q',q}|^2.
\]
It follows immediately from Proposition \ref{prop:Hoeffding:vec} and \eqref{Rademacher:Hoeffding} that the
increments of the process satisfy
\begin{align}\label{YM:increments}
\P(\|\mtx{B}(\vct{x})\vct{\epsilon} - \mtx{B}(\vct{x}')\vct{\epsilon}\|_2 \geq u \|\mtx{B}(\vct{x}) - \mtx{B}(\vct{x}')\|_F) \leq 8 e^{-u^2/16}.
\end{align}
This allows to apply the following variant of Dudley's inequality for vector-valued processes
in $\ell_2$.
\begin{theorem}\label{Dudley:l2} Let $\vct{R}_x$, $\vct{x} \in T$, be a process with values in $\C^\m$ indexed by a metric space
$(T,d)$, with increments that satisfy the subgaussian tail estimate
\[
\P(\|\vct{R}_{\vct{x}} - \vct{R}_{\vct{x'}}\|_2 \geq u d(\vct{x},\vct{x'}) ) \leq 8 e^{-u^2/16}.
\]
Then, for an arbitrary $\vct{x_0} \in T$ and a universal constant $K>0$,
\begin{align}
 \Big(\E \sup_{\vct{x} \in T} \| \vct{R}_{\vct{x}} - \vct{R}_{\vct{x_0}}\|_2^2\Big)^{1/2} \leq K
\int_0^\infty \sqrt{\log(N(T,d,u))} du, \label{eqn:increments}
\end{align}
where $N(T,d,u)$ denote the covering numbers of $T$ with respect to $d$ and radius $u>0$.
\end{theorem}
\begin{proof} The proof follows literally
the lines of the standard proof of Dudley's inequalities for
scalar-valued subgaussian processes, see for instance
\cite[Theorem 6.23]{ra10} or \cite{azws09,leta91,ta05-2}. One only has to replace
the triangle inequality for the absolute value by the one for $\|\cdot\|_2$ in $\C^\m$. 
\end{proof}

We have $d=d_2$ defined above, and, hence, \eqref{entropy:int1a} provides
us with the right hand side of \eqref{eqn:increments}. Using the fact that
here, $\vct{R}_x=\mtx{B}(\vct{x})\vct{\epsilon}$,  we conclude that
\begin{align}
V&=    \E \sup_{\vct{x} \in T_s} \| \mtx{B}(\vct{x})\vct{\epsilon} \|_2^2=\E \sup_{\vct{x} \in T_s} \| \mtx{B}(\vct{x})\vct{\epsilon}-\mtx{B}(\vct{0})\vct{\epsilon} \|_2^2\notag\\
             &   \leq \big(KC \sqrt{n s^{3/2}} \sqrt{\log(n)} \log(\sparsity)\big)^2 
             \leq C'  n s^{3/2}
 \log(n) \log^2(\sparsity).\notag
\end{align}

Plugging these estimates into \eqref{eq:boucheron} and simplifying
leads to our result, compare with \cite{rarotr10}.
In particular, Theorem \ref{theorem:main}(b) follows.

\section*{acknoeldegements}
G\"otz E. Pfander appreciates the support by the {\it Deutsche Forschungsgemeinschaft} (DFG) under grant 50292 DFG PF-4 Sampling Operators. Holger Rauhut acknowledges generous support by
the Hausdorff Center for Mathematics, and funding by the Starting Independent Researcher Grant
StG-2010 258926-SPALORA from the European Research Council (ERC). Joel A. Tropp was supported in part by the Defense Advanced Research Projects Agency (DARPA) and the Office of Naval Research (ONR) under Grants N66001-06-1-2011 and N66001-08-1-2065.


\begin{thebibliography}{10}
\providecommand{\url}[1]{{#1}}
\providecommand{\urlprefix}{URL }
\expandafter\ifx\csname urlstyle\endcsname\relax
  \providecommand{\doi}[1]{DOI~\discretionary{}{}{}#1}\else
  \providecommand{\doi}{DOI~\discretionary{}{}{}\begingroup
  \urlstyle{rm}\Url}\fi

\bibitem{al80}
{A}lltop, W.O.: {C}omplex sequences with low periodic correlations.
\newblock {I}{E}{E}{E} {T}rans. {I}nform. {T}heory \textbf{26}(3), 350--354
  (1980)

\bibitem{azws09}
{A}zais, J.M., {W}schebor, M.: {L}evel {S}ets and {E}xtrema of {R}andom
  {P}rocesses and {F}ields.
\newblock {J}ohn {W}iley \& {S}ons {I}nc. (2009)

\bibitem{badadewa08}
{B}araniuk, R.G., {D}avenport, M., {D}e{V}ore, R.A., {W}akin, M.: {A} simple
  proof of the restricted isometry property for random matrices.
\newblock {C}onstr. {A}pprox. \textbf{28}(3), 253--263 (2008)

\bibitem{be63}
{B}ello, P.A.: {C}haracterization of {R}andomly {T}ime-{V}ariant {L}inear
  {C}hannels.
\newblock {I}{E}{E}{E} {T}rans. {C}omm. \textbf{11}, 360--393 (1963)

\bibitem{blda09}
{B}lumensath, T., {D}avies, M.: {I}terative hard thresholding for compressed
  sensing.
\newblock {A}ppl. {C}omput. {H}armon. {A}nal. \textbf{27}(3), 265--274 (2009)

\bibitem{boluma03}
{B}oucheron, S., {L}ugosi, G., {M}assart, P.: {C}oncentration inequalities
  using the entropy method.
\newblock {A}nn. {P}robab. \textbf{31}(3), 1583--1614 (2003)

\bibitem{bu01}
{B}uchholz, A.: {O}perator {K}hintchine inequality in non-commutative
  probability.
\newblock {M}ath. {A}nn. \textbf{319}, 1--16 (2001)

\bibitem{cawaxu10}
{C}ai, T., {W}ang, L., {X}u, G.: {S}hifting inequality and recovery of sparse
  vectors.
\newblock {I}{E}{E}{E} {T}rans. {S}ignal {P}rocess. \textbf{58}(3), 1300--1308
  (2010)

\bibitem{ca06-1}
{C}and{\`e}s, E.J.: {C}ompressive sampling.
\newblock In: {P}roceedings of the {I}nternational {C}ongress of
  {M}athematicians. {M}adrid, {S}pain (2006)

\bibitem{ca08}
{C}and{\`e}s, E.J.: {T}he restricted isometry property and its implications for
  compressed sensing.
\newblock preprint  (2008)

\bibitem{carota06}
{C}and{\`e}s, E.J., {J}., {T}ao, T., {R}omberg, J.: {R}obust uncertainty
  principles: exact signal reconstruction from highly incomplete frequency
  information.
\newblock {I}{E}{E}{E} {T}rans. {I}nform. {T}heory \textbf{52}(2), 489--509
  (2006)

\bibitem{carota06-1}
{C}and{\`e}s, E.J., {R}omberg, J., {T}ao, T.: {S}table signal recovery from
  incomplete and inaccurate measurements.
\newblock {C}omm. {P}ure {A}ppl. {M}ath. \textbf{59}(8), 1207--1223 (2006)

\bibitem{cata06}
{C}and{\`e}s, E.J., {T}ao, T.: {N}ear optimal signal recovery from random
  projections: universal encoding strategies?
\newblock {I}{E}{E}{E} {T}rans. {I}nform. {T}heory \textbf{52}(12), 5406--5425
  (2006)

\bibitem{ca85}
{C}arl, B.: {I}nequalities of {B}ernstein-{J}ackson-type and the degree of
  compactness of operators in {B}anach spaces.
\newblock {A}nn. {I}nst. {F}ourier ({G}renoble) \textbf{35}(3), 79--118 (1985)

\bibitem{chdosa99}
{C}hen, S.S., {D}onoho, D.L., {S}aunders, M.A.: {A}tomic decomposition by
  {B}asis {P}ursuit.
\newblock {S}{I}{A}{M} {J}. {S}ci. {C}omput. \textbf{20}(1), 33--61 (1999)

\bibitem{Chr03}
Christensen, O.: An {I}ntroduction to {F}rames and {R}iesz {B}ases.
\newblock Applied and Numerical Harmonic Analysis. Birkh\"auser Boston Inc.,
  Boston, MA (2003)

\bibitem{codade09}
{C}ohen, A., {D}ahmen, W., {D}e{V}ore, R.A.: {C}ompressed sensing and best
  k-term approximation.
\newblock {J}. {A}mer. {M}ath. {S}oc. \textbf{22}(1), 211--231 (2009)

\bibitem{Cor01}
Correia, L.M.: Wireless Flexible Personalized Communications.
\newblock John Wiley \& Sons, Inc., New York, NY, USA (2001)

\bibitem{do06-2}
{D}onoho, D.L.: {C}ompressed sensing.
\newblock {I}{E}{E}{E} {T}rans. {I}nform. {T}heory \textbf{52}(4), 1289--1306
  (2006)

\bibitem{dota09}
{D}onoho, D.L., {T}anner, J.: {C}ounting faces of randomly-projected polytopes
  when the projection radically lowers dimension.
\newblock {J}. {A}mer. {M}ath. {S}oc. \textbf{22}(1), 1--53 (2009)

\bibitem{fora11}
{F}ornasier, M., {R}auhut, H.: {C}ompressive sensing.
\newblock In: O.~{S}cherzer (ed.) {H}andbook of {M}athematical {M}ethods in
  {I}maging, pp. 187--228. {S}pringer (2011)

\bibitem{fo10-3}
{F}oucart, S.: {A} note on guaranteed sparse recovery via
  $\ell_1$-minimization.
\newblock {A}ppl. {C}omput. {H}armon. {A}nal. \textbf{29}(1), 97--103 (2010)

\bibitem{fo10-2}
{F}oucart, S.: {H}ard thresholding pursuit: an algorithm for compressive
  sensing.
\newblock preprint  (2010)

\bibitem{fo10}
{F}oucart, S.: {S}parse recovery algorithms: sufficient conditions in terms of
  restricted isometry constants.
\newblock In: {P}roceedings of the 13th {I}nternational {C}onference on
  {A}pproximation {T}heory (2010)

\bibitem{foparaul10}
{F}oucart, S., {P}ajor, A., {R}auhut, H., {U}llrich, T.: {T}he {G}elfand widths
  of $\ell_p$-balls for $0 < p \leq 1$.
\newblock {J}. {C}omplexity \textbf{26}(6), 629--640 (2010)

\bibitem{gagl84}
{G}arnaev, A., {G}luskin, E.: {O}n widths of the {E}uclidean ball.
\newblock {S}ov. {M}ath., {D}okl. \textbf{30}, 200--204 (1984)

\bibitem{grpf08}
{G}rip, N., {P}fander, G.: A discrete model for the efficient analysis of
  time-varying narrowband communication channels.
\newblock Multidim. Syst. Signal Processing \textbf{19}(1), 3--40 (2008)

\bibitem{Gro01}
Gr{\"o}chenig, K.: Foundations of Time-Frequency Analysis.
\newblock Applied and Numerical Harmonic Analysis. Birkh{\"a}user, Boston, MA
  (2001)

\bibitem{bahanora10}
{H}aupt, J., {B}ajwa, W., {R}az, G., {N}owak, R.: {T}oeplitz compressed sensing
  matrices with applications to sparse channel estimation.
\newblock {I}{E}{E}{E} {T}rans. {I}nform. {T}heory \textbf{56}(11), 5862--5875
  (2010)

\bibitem{hest09}
{H}erman, M., {S}trohmer, T.: {H}igh-resolution radar via compressed sensing.
\newblock {I}{E}{E}{E} {T}rans. {S}ignal {P}rocess. \textbf{57}(6), 2275--2284
  (2009)

\bibitem{krpfra07}
Krahmer, F., Pfander, G.E., Rashkov, P.: Uncertainty in time-frequency
  representations on finite abelian groups and applications.
\newblock Appl. Comput. Harmon. Anal. \textbf{25}(2), 209--225 (2008) 

\bibitem{LPW05}
Lawrence, J., Pfander, G., Walnut, D.: Linear independence of {G}abor systems
  in finite dimensional vector spaces.
\newblock J. Fourier Anal. Appl. \textbf{11}(6), 715--726 (2005)

\bibitem{leta91}
{L}edoux, M., {T}alagrand, M.: {P}robability in {B}anach spaces.
\newblock {S}pringer-{V}erlag, {B}erlin, {H}eidelberg, {N}ew{Y}ork (1991)

\bibitem{mepato09}
{M}endelson, S., {P}ajor, A., {T}omczak {J}aegermann, N.: {U}niform uncertainty
  principle for {B}ernoulli and subgaussian ensembles.
\newblock {C}onstr. {A}pprox. \textbf{28}(3), 277--289 (2009)

\bibitem{mi87}
Middleton, D.: Channel modeling and threshold signal processing in underwater
  acoustics: An analytical overview.
\newblock IEEE J. Oceanic Eng. \textbf{12}(1), 4--28 (1987)

\bibitem{na95}
{N}atarajan, B.K.: {S}parse approximate solutions to linear systems.
\newblock {S}{I}{A}{M} {J}. {C}omput. \textbf{24}, 227--234 (1995)

\bibitem{neve09-1}
{N}eedell, D., {V}ershynin, R.: {U}niform uncertainty principle and signal
  recovery via regularized orthogonal matching pursuit.
\newblock {F}ound. {C}omput. {M}ath. \textbf{9}(3), 317--334 (2009)

\bibitem{Pae01}
P\"atzold, M.: Mobile Fading Channels: Modelling, Analysis and Simulation.
\newblock John Wiley \& Sons, Inc. (2001)

\bibitem{gide99}
de~la {P}e{\~n}a, V., {G}in{\'e}, E.: {D}ecoupling. {F}rom {D}ependence to
  {I}ndependence.
\newblock {P}robability and its {A}pplications ({N}ew {Y}ork).
  {S}pringer-{V}erlag (1999)

\bibitem{PR09}
Pfander, G., Rauhut, H.: Sparsity in time--frequency representations.
\newblock J. Fourier Anal. Appl. \textbf{16}(2), 233--260 (2010)

\bibitem{pfrata08}
{P}fander, G.E., {R}auhut, H., {T}anner, J.: {I}dentification of matrices
  having a sparse representation.
\newblock {I}{E}{E}{E} {T}rans. {S}ignal {P}rocess. \textbf{56}(11), 5376--5388
  (2008)

\bibitem{ra08}
{R}auhut, H.: {S}tability results for random sampling of sparse trigonometric
  polynomials.
\newblock {I}{E}{E}{E} {T}rans. {I}nformation {T}heory \textbf{54}(12),
  5661--5670 (2008)

\bibitem{ra09}
{R}auhut, H.: {C}irculant and {T}oeplitz matrices in compressed sensing.
\newblock In: {P}roc. {S}{P}{A}{R}{S}'09 (2009)

\bibitem{ra10}
{R}auhut, H.: {C}ompressive {S}ensing and {S}tructured {R}andom {M}atrices.
\newblock In: M.~{F}ornasier (ed.) {T}heoretical {F}oundations and {N}umerical
  {M}ethods for {S}parse {R}ecovery, \emph{{R}adon {S}eries {C}omp. {A}ppl.
  {M}ath.}, vol.~9, pp. 1--92. de{G}ruyter (2010)

\bibitem{pfra10}
{R}auhut, H., {P}fander, G.E.: {S}parsity in time-frequency representations.
\newblock {J}. {F}ourier {A}nal. {A}ppl. \textbf{16}(2), 233--260 (2010)

\bibitem{rarotr10}
{R}auhut, H., {R}omberg, J., {T}ropp, J.: Restricted isometries for partial
  random circulant matrices.
\newblock Appl. Comput. Harmonic Anal.  (to appear).
\newblock DOI:10.1016/j.acha.2011.05.001

\bibitem{rascva08}
{R}auhut, H., {S}chnass, K., {V}andergheynst, P.: {C}ompressed sensing and
  redundant dictionaries.
\newblock {I}{E}{E}{E} {T}rans. {I}nform. {T}heory \textbf{54}(5), 2210 -- 2219
  (2008)

\bibitem{rawa10}
{R}auhut, H., {W}ard, R.: {S}parse {L}egendre expansions via
  $l_1$-minimization.
\newblock preprint  (2010)

\bibitem{ruve08}
{R}udelson, M., {V}ershynin, R.: {O}n sparse reconstruction from {F}ourier and
  {G}aussian measurements.
\newblock {C}omm. {P}ure {A}ppl. {M}ath. \textbf{61}, 1025--1045 (2008)

\bibitem{st99}
Stojanovic, M.: Underwater acoustic communications.
\newblock In: J.G. Webster (ed.) Encyclopedia of Electrical and Electronics
  Engineering, vol.~22, pp. 688--698. John Wiley {\&} Sons (1999)

\bibitem{hest03}
{S}trohmer, T., {H}eath, R.W.j.: {G}rassmannian frames with applications to
  coding and communication.
\newblock {A}ppl. {C}omput. {H}armon. {A}nal. \textbf{14}(3), 257--275 (2003)

\bibitem{ta96-2}
{T}alagrand, M.: {N}ew concentration inequalities in product spaces.
\newblock {I}nvent. {M}ath. \textbf{126}(3), 505--563 (1996)

\bibitem{ta05-2}
{T}alagrand, M.: {T}he {G}eneric {C}haining.
\newblock {S}pringer {M}onographs in {M}athematics. {S}pringer-{V}erlag (2005)

\bibitem{netr08}
{T}ropp, J., {N}eedell, D.: {C}o{S}a{M}{P}: {I}terative signal recovery from
  incomplete and inaccurate samples.
\newblock {A}ppl. {C}omput. {H}armon. {A}nal. \textbf{26}(3), 301--321 (2008)

\bibitem{tr04}
{T}ropp, J.A.: {G}reed is good: {A}lgorithmic results for sparse approximation.
\newblock {I}{E}{E}{E} {T}rans. {I}nform. {T}heory \textbf{50}(10), 2231--2242
  (2004)

\end{thebibliography}

\end{document}